\documentclass[11pt,a4paper]{article}
\usepackage{vmargin}
\usepackage{amsfonts,amsmath,amssymb}
\usepackage{amsthm}
\usepackage{graphicx}
\graphicspath{{pictures/}}
\usepackage{color}

\newtheorem{theorem}{Theorem}
\newtheorem{lemma}{Lemma}
\newtheorem{corollary}{Corollary}
\newtheorem{conjecture}{Conjecture}

\usepackage{pifont}
\usepackage{amsmath}
\usepackage{amsfonts,amsmath,amssymb}
\usepackage{graphicx}
\usepackage{amsbsy}
\graphicspath{{pictures/}}
\usepackage{amsmath}
\usepackage{color}
\usepackage{todonotes}
\usepackage{algorithm}
\usepackage{algpseudocode}
\usepackage{comment}
\usepackage{mathtools}
\usepackage{complexity}
\usepackage{subcaption}
\usepackage{hyperref}
\usepackage{enumitem}
\usepackage{footnote}
\usepackage{todonotes}
\usepackage{comment}
\usepackage{amsmath}
\usepackage{hyperref}
\usepackage{pgfplots}

\title{From Hotelling to Load Balancing: Approximation and the
	Principle of Minimum Differentiation\\{\small (full version)}}

\author{Matthias Feldotto\thanks{Heinz Nixdorf Institute \& Departement of Computer Science, Paderborn University, Paderborn, Germany \texttt{feldi@mail.upb.de}} \and Pascal Lenzner\thanks{Algorithm Engineering Group, Hasso Plattner Institute, Potsdam, Germany, \texttt{\{pascal.lenzner,louise.molitor\}@hpi.de}} \and Louise Molitor\footnotemark[2] \and Alexander Skopalik\thanks{Faculty of Electrical Engineering, Mathematic \& Computer Science, University of Twente, Enschede, The Netherlands \texttt{a.skopalik@utwente.nl}}}

\date{~}
\allowdisplaybreaks %page breaks in math environments

\newcommand{\naturals}{\ensuremath{\mathbb{N}}}

\newcommand{\player}{\ensuremath{i}}
\newcommand{\players}{\ensuremath{\mathcal{N}}}
\newcommand{\strategies}{\ensuremath{S}}
\newcommand{\strategiesPlayer}{\ensuremath{\strategies_i}}
\newcommand{\strategy}{\ensuremath{s}}
\newcommand{\strategyPlayer}{\ensuremath{\strategy_i}}
\newcommand{\strategyProfile}{\ensuremath{\textsf{s}}}
\newcommand{\strategyProfiles}{\ensuremath{\textsf{S}}}

\newcommand{\utility}{\ensuremath{u}}
\newcommand{\utilityPlayer}{\ensuremath{\utility_i}}
\newcommand{\socialcosts}{\ensuremath{SC}}
\newcommand{\quality}{\ensuremath{Q}}

\newcommand{\clients}{\ensuremath{Z}}
\newcommand{\client}{\ensuremath{z}}
\newcommand{\choiceFunction}{\ensuremath{f}}

\newcommand{\choiceFunctions}{\ensuremath{\textsf{f}}}
\newcommand{\costsCl}{\ensuremath{C}}
\newcommand{\costsClient}{\ensuremath{\costsCl_{\client}}}
\newcommand{\clientStrategyWithout}{\ensuremath{t}}
\newcommand{\clientStrategy}{\ensuremath{\clientStrategyWithout_{\client}}}

\newcommand{\clientStrategyProfiles}{\ensuremath{\textsf{F}}}
\newcommand{\clientStrategyProfile}{\ensuremath{\choiceFunctions}}
\newcommand{\spe}{\ensuremath{\mathrm{SPE}}}

\newcommand{\approxSpe}{\ensuremath{\rho\mathrm{-\spe}}}

\newcommand{\load}{\ensuremath{\ell}}
\newcommand{\loadPlayer}{\ensuremath{\load_i}}

\newcommand{\nal}{\ensuremath{\left(1-\alpha\right)}}
\newcommand{\al}{\ensuremath{\alpha}}

\newcommand{\approxPairing}{\ensuremath{\psi^\text{pair}}}
\newcommand{\approxPairingGeneral}{\ensuremath{\approxPairing_{n,\alpha}}}

\newcommand{\approxOptimum}{\ensuremath{\psi^\text{opt}}}
\newcommand{\approxOptimumGeneral}{\ensuremath{\approxOptimum_{n,\alpha}}}

\newcommand{\potentialFunction}{\ensuremath{\Phi}}
\newcommand{\contFraQua}{\ensuremath{\tilde{K}}}
\newcommand{\contFraLin}{\ensuremath{\hat{K}}}
\def\contFracOpe{%
	\operatornamewithlimits{%
		\mathchoice{% * Display style
			\vcenter{\hbox{\huge $\mathcal{K}$}}%
		}{%           * Text style
			\vcenter{\hbox{\Large $\mathcal{K}$}}%
		}{%           * Script style
			\mathrm{\mathcal{K}}%
		}{%           * Script script style
			\mathrm{\mathcal{K}}%
		}
	}
}

\begin{document}

\maketitle

\begin{abstract}
 \noindent Competing firms tend to select similar locations for their stores. This phenomenon, called the principle of minimum differentiation, was captured by Hotelling with a landmark model of spatial competition but is still the object of an ongoing scientific debate. Although consistently observed in practice, many more realistic variants of Hotelling’s model fail to support minimum differentiation or do not have pure equilibria at all. In particular, it was recently proven for a generalized model which incorporates negative network externalities and which contains Hotelling’s model and classical selfish load balancing as special cases, that the unique equilibria do not adhere to minimum differentiation. Furthermore, it was shown that for a significant parameter range pure equilibria do not exist. We derive a sharp contrast to these previous results by investigating Hotelling’s model with negative network externalities from an entirely new angle: approximate pure subgame perfect equilibria. This approach allows us to prove analytically and via agent-based simulations that approximate equilibria having good approximation guarantees and that adhere to minimum differentiation exist for the full parameter range of the model. Moreover, we show that the obtained approximate equilibria have high social welfare.
 \end{abstract}
 
 \section{Introduction}
 
 The choice of a profitable facility location is one of the core strategic decisions for firms competing in a spatial market. Finding the right location is a classical object of research and has kindled the rich and interdisciplinary research area called Location Analysis~\cite{DBLP:journals/eor/RevelleE05, DBLP:journals/transci/EiseltLT93, doi:10.1002/9780470400531.eorms0477}. In this paper we investigate one of the landmark models of spatial competition and strategic product differentiation where facilities offering the same service for the same price compete in a linear spatial market. Originally introduced by Hotelling~\cite{10.2307/2224214} and later extended by Downs~\cite{doi:10.1086/257897} to model political competition, the model is usually referred to as the \emph{Hotelling-Downs model}. It assumes a market of infinitely many clients which are distributed evenly on a line and finitely many firms which want to open a facility and which strategically select a specific facility location in the market to sell their service. Every client wants to obtain the offered service and selects the nearest facility to get it. The utility of the firms is proportional to the number of clients visiting their facility. Thus, the location decision of a firm depends on the facility locations of all its competitors as well as on the anticipated behavior of the clients. This two-stage setting is challenging to analyze but at the same time yields a plausible prediction of real-world phenomena. 
 
 One such phenomenon is known as the \emph{principle of minimum differentiation}~\cite{boulding1941economic,Eaton1975} and it states that competing firms selling the same service tend to co-locate their facilities instead of spreading them evenly along the market. This can be readily observed in practice, e.g., stores of different fast-food chains or consumer goods shops are often located right next to each other. For the original version where clients simply select the nearest facility, Eaton and Lipsey~\cite{Eaton1975} proved in a seminal paper that for $n\neq 3$ competing firms the Hotelling-Downs model has pure subgame perfect equilibria which respect the principle of minimum differentiation. 
 
 However, the original Hotelling-Downs model is overly simple and lacks crucial properties found in practice. For example it does not incorporate negative network externalities for the clients. When choosing which facility to patronize real clients would not only evaluate distances but also how congested a facility is. Many other clients visiting  the same facility induce a higher waiting time to get serviced and thus it may be better for a client to select a different facility which is farther away but has fewer clients. This natural and more realistic variant, where the cost function of a client is a linear combination of distance and waiting time, was proposed by Kohlberg~\cite{KOHLBERG1983211} and will be the focus of our attention. Kohlberg's model is especially interesting, since it can be interpreted as an interpolation of two extreme models: the Hotelling-Downs model, where clients select the nearest facility and classical \emph{Selfish Load Balancing}~\cite{vocking2007selfish}, where clients select the least congested facility.
 
 For Kohlberg's model it is known that no pure subgame perfect equilibria exist where the facility locations are pairwise different. Furthermore, in a recent paper Peters et al.~\cite{Peters2018} show for up to six facilities that pure equilibria exist if and only if there is an even number of facilities and the clients' cost function is tilted heavily towards preferring less congested facilities. Moreover, in sharp contrast to the principle of minimum differentiation, they show that in these unique equilibria only two facilities are co-located. 
 
 In this paper we re-establish the principle of minimum differentiation for Kohlberg's model by considering \emph{approximate pure subgame perfect equilibria}. We show analytically and by extensive agent-based simulations that for any client cost function which is a linear combination of distances and congestion approximate subgame perfect equilibria exist which respect the principle of minimum differentiation and where each firm can only increase her utility by a small multiplicative factor by deviating to another facility location. Moreover, we show that the obtained approximate equilibria are also close to optimal for the whole society of clients.
 
We believe that in contrast to studying exact subgame perfect equilibria, investigating approximate subgame perfect equilibria yields more reliable predictions since the study of exact equilibria assumes actors who radically change their current strategy even if they can improve only by a tiny margin. In the real world this is not true, as many actors only move out of their ”comfort zone” if a significant improvement can be realized. This threshold behavior can naturally be modeled via a suitably chosen approximation factor. Furthermore, approximate equilibria are the only hope for a plausible prediction for many variants of the Hotelling-Downs model where exact equilibria do not exist. To the best of our knowledge, approximate equilibria have not been studied before in the realm of Location Analysis.
 
 \subsection{Related Work}
 The Hotelling-Downs model was also analyzed for non-linear markets, e.g., on graphs~\cite{palvolgyi2011hotelling,DBLP:journals/corr/FournierS16,DBLP:journals/corr/Fournier16},  fixed locations on a circle~\cite{10.2307/3003323}, finite sets of locations~\cite{Nunez2016,Nunez2017}, and optimal interval division~\cite{tian2015optimal}. Moreover, many facility location games are variants of the Hotelling-Downs model and there is a rich body of work analyzing competitive facility location in the non-cooperative setting, e.g. Vetta~\cite{DBLP:conf/focs/Vetta02}, Cardinal \& Hoefer~\cite{DBLP:journals/tcs/CardinalH10}, Saban \& Stier-Moses~\cite{DBLP:conf/wine/SabanM12}, Drees et al.~\cite{drees2016strategic}, and in the cooperative setting, e.g. Goemans \& Skutella~\cite{DBLP:journals/jal/GoemansS04}. 
 Additionally, Fotakis \& Tzamos~\cite{fotakis2014power} and Feldman et al.~\cite{feldman2016voting} investigate facility location mechanisms. Another closely related family of games is the class of Voronoi games~\cite{DBLP:journals/tcs/AhnCCGO04,DBLP:conf/esa/DurrT07,DBLP:journals/tcs/BandyapadhyayBDS15}. We phrase our model in terms of a facility location game, but, in contrast to the above works on facility location and Voronoi games our clients do not necessarily select the nearest facility to get serviced.
 
 % load models
 Kohlberg~\cite{KOHLBERG1983211} originally defined the cost for a client at location $x$ to visit a facility at location $y$ given some client distribution as the sum of the distance between $x$ and $y$ on the line and $\gamma$ times the number of clients currently served by the chosen facility. Moreover, he claims that no subgame perfect equilibrium exists for more than two facilities. This claim was later refuted by Peters et al.~\cite{Peters2018} who proved the existence of a unique subgame perfect equilibrium for $n=4$ and $n=6$ for large values of the parameter $\gamma$. Moreover, they conjecture that a unique subgame perfect equilibrium exists for any even number of facilities if $\gamma$ is sufficiently large and they give the corresponding equilibrium candidate. In sharp contrast to the principle of minimum differentiation, the equilibrium candidate exhibits only two facilities which are co-located.
 Additionally, they investigate an asymmetric variant of this model, where the waiting time of each facility can be different. 
 One extreme case of Kohlberg's model is the setting in which the clients are only interested in selecting the least congested facility independent of its distance. 
 This setting is captured by simple load balancing games~\cite{vocking2007selfish} and it is easy to see that in this case any location vector of the facilities must be a subgame perfect equilibrium. 
 
 In our model facilities offer their service for the same price. Models where facilities can also strategically set the price have been considered~\cite{10.2307/20075765, navon1995product, GRILO2001385, 10.2307/1911955, osborne1986nature}.   
 Setting different prices under network externalities was investigated by Heikkinen~\cite{DBLP:journals/or/Heikkinen14}. Moreover, Ahlin \& Ahlin~\cite{ahlin2013product} show in a version with pricing that negative network externalities lead to less differentiation between the facilities.  
 
 Other recent work investigates different client attraction functions instead of simply using the distance to the facilities. Ben-Porat \& Tennenholtz~\cite{DBLP:conf/wine/Ben-PoratT17} use a connection to the Shapley value to show the existence of pure equilibria through a potential function argument.
 Whereas Feldman et al.~\cite{DBLP:conf/aaai/FeldmanFO16} consider the case where facilities have a limited attraction interval and the uniformly distributed clients decide randomly which facility to choose if attracted by more than one facility. Interestingly they prove that pure Nash equilibria exist and that the Price of Anarchy is low. Later Shen \& Wang~\cite{shen2017hotelling} generalized the model to arbitrary client distributions.
 
 Using agent-based simulations for variants of the Hotelling-Downs model seems to be a quite novel approach. We could find only the recent work of van Leeuwen \& Lijesen~\cite{van2016agents} in which the authors claim to present the first such approach. They study a multi-stage variant with pricing which is different from our setting. 
 
 \subsection{Our Contribution}
 We study approximate pure subgame perfect equilibria in Kohlberg's model of spatial competition with negative network externalities in which $n$ facility players strategically select a location in a linear market. Our slightly reformulated model has a parameter $0\leq \alpha \leq 1$, where $\alpha = 0$ yields the original Hotelling-Downs model, i.e., clients select the nearest facility, and where $\alpha = 1$ yields classical selfish load balancing, i.e., clients select the least congested facility.  
 
 First, we study the case $n=3$, which for $\alpha = 0$ is the famous unique case of the Hotelling-Downs model where exact equilibria do not exist. We show that for all $\alpha$ an approximate subgame perfect equilibrium exists with approximation factor $\rho \leq 1.2808$. Moreover, for $\alpha = 0$ we show that this bound is tight. 
 
 Next, we consider the facility placement which is socially optimal for the clients and analyze its approximation factor, i.e., we answer the question how tolerant the facility players have to be to accept the social optimum placement for the clients. For this placement, in which the facilities are uniformly distributed along the linear market, we derive exact analytical results for $4\leq n \leq 10$. Building on this and on a conjecture specifying the facility which has the best improving deviation, we generalize our results to $n\geq 4$. We find that the obtained approximation factor $\rho$ approaches $1.5$ for low $\alpha$ which implies that in these cases facility players must be very tolerant to support these client optimal placements. 
 
 We contrast this by our main contribution, which is a study of a facility placement proposed by Eaton \& Lipsey~\cite{Eaton1975} from an approximation perspective. This placement supports the principle of minimum differentiation since all but at most one facility are co-located with another facility and at the same time it is an exact equilibrium for both extreme cases of the model, i.e., for $\alpha = 0$ and $\alpha =1$. We provide analytical proofs that for these placements $\rho \leq 1.0866$ holds for $4\leq n \leq 10$. Moreover, based on another conjecture, we show that for arbitrary even $n\geq 10$ we get $\rho \approx 1.08$. 
 
 Our conjectures used for proving the general results are based on the analytical results for $n\leq 10$ and on extensive agent-based simulations of a discretized variant of the model. It turns out that these simulations yield reliable predictions for the original model and we also use them for providing promising results for the general case with odd $n$. In particular, we demonstrate that empirically we have $\rho \approx 1.08$ for arbitrary $n\geq 10$.   
 
 Last but not least, we show that the facility placements proposed by Eaton \& Lipsey~\cite{Eaton1975} are also socially good for the clients. We compare their social cost with the cost of the social optimum placement and prove a low ratio for all $\alpha$.  
 
 Overall, we prove that for Kohlberg's model facility placements exist which 
 \begin{itemize}
 	\item[(1)] adhere to the principle of minimum differentiation,
 	\item[(2)] are close to stability in the sense that facilities can only improve their utility by at most $8\%$ by deviating and
 	\item[(3)] these placements are also socially beneficial for all clients.
 \end{itemize}    

\section{Model and Preliminaries}\label{model}
We model the scenario as a two-stage game with two types of players, a set of \emph{facilities} $\players$ each offering the same service for the same price and a set of
\emph{clients} $\clients$ each choosing one facility to get serviced from.
There are $n$ facility players $\players= \{1, \dots,n\}$, which choose a location in the interval $\strategies = [0,1]$.
We denote a strategy vector for the facility players as $\strategyProfile = (\strategy_1, \ldots, \strategy_n)$, where $\strategyPlayer \in \strategies$ denotes the chosen location of facility player $i$.
For notational purposes, $(\strategyProfile_{-i}, \strategyPlayer')$ denotes the vector that results when player $i$ changes her strategy in $\strategyProfile$ from $\strategy_i$ to $\strategy_i'$.
%client strategies
For the clients, we consider a continuum of infinitely many clients represented by the interval $\clients = [0,1]$.
Every point $\client \in \clients$ corresponds to a client that chooses a facility $i\in \players$ to get serviced.
Hence, the strategy space $\strategies_{\client}$ of a client $\client \in \clients$ is the set of facilities, i.e., $\strategies_{\client} = \players = \{1, \dots,n\}$, with $\clientStrategy \in \players$ being the current strategy selection.
We define $\choiceFunctions: \strategyProfiles \times \clients \to \players$ as the mapping induced by the clients' facility choices.
Given a facility location vector $\strategyProfile$, a client $\client \in \clients$ selects the facility $\choiceFunctions(\strategyProfile, \client)$.
To express strategy changes of single agents, we define by $(\choiceFunctions_{-\client}, \choiceFunction'_{\client})$ the choice function which results, if only the mapping of the agent at $\client$ changes from the value $\choiceFunction(\client)$ to the value $\choiceFunction'(\client)$.
The set of all possible client agent strategy profiles is given by $\clientStrategyProfiles = \players^{\strategyProfiles \times \clients}$.
A strategy profile (of the facilities and clients) is a pair $(\strategyProfile, \choiceFunctions) \in \strategyProfiles \times \clientStrategyProfiles$, where $\strategyProfile$ is the vector of strategies of the facility players and $\choiceFunctions$ is the choice function determining the strategies of the client agents.

To measure how many clients select a specific strategy, we consider only client choice functions $\choiceFunctions$, where the interval $\clients$ is partitioned into $n$ finite sets of intervals $\mathcal{J}_1(\strategyProfile, \choiceFunctions),\mathcal{J}_2(\strategyProfile, \choiceFunctions),$ $\ldots, \mathcal{J}_{|\players|}(\strategyProfile, \choiceFunctions)$, where $\mathcal{J}_i(\strategyProfile, \choiceFunctions) = \{\mathcal{I}_i^1(\strategyProfile, \choiceFunctions),\dots,\mathcal{I}_i^{k_i}(\strategyProfile, \choiceFunctions)\}$, for some $k_i$, with disjoint intervals $\mathcal{I}_i^j \subseteq [0,1]$ and such that for all clients $\client \in \mathcal{I}_i^j(\strategyProfile, \choiceFunctions)\ \forall j \in \{1, \ldots, k_i\}$ we have $\choiceFunctions(\strategyProfile,\client) = i$. We call such client mappings \emph{measurable mappings}.

Given a measurable client mapping $\choiceFunctions$ and the corresponding induced partition into $n$ finite sets of intervals $\mathcal{J}_1(\strategyProfile, \choiceFunctions), \dots, \mathcal{J}_n(\strategyProfile, \choiceFunctions)$ where $|\mathcal{I}_i^j(\strategyProfile, \choiceFunctions)|$ is the length of interval $\mathcal{I}_i^j(\strategyProfile, \choiceFunctions)$. We define the \emph{load of facility $i$} as
$$\loadPlayer(\strategyProfile, \clientStrategyProfile) = \sum_{\mathcal{I}_i^j(\strategyProfile, \choiceFunctions) \in \mathcal{J}_i(\strategyProfile, \choiceFunctions)} |\mathcal{I}_i^j(\strategyProfile, \choiceFunctions)|.$$

%client costs
Given a facility location vector $\strategyProfile$ and a measurable client mapping $\choiceFunctions$, the \emph{cost} $\costsClient$ of a single client at some point $\client\in \clients$ is proportional to her distance from her chosen facility $\choiceFunctions(\strategyProfile, \client)$ and the current load $\load_{\choiceFunctions(\strategyProfile,\client)}(\strategyProfile,\choiceFunctions)$ of that facility.
The relative influence of these two objectives is adjusted via the parameter $\alpha\in  [0,1]$.
Thus, the cost of a client at point $\client \in \clients$ is 
$$\costsClient(\strategyProfile,\clientStrategyProfile) = (1-\alpha) \cdot |\strategy_{\choiceFunctions(\strategyProfile,\client)} - \client| + \alpha \cdot \load_{\choiceFunctions(\strategyProfile,\client)}(\strategyProfile, \clientStrategyProfile).$$

For $\alpha=0$, where clients simply ignore the facility loads, this corresponds to the client cost function from Hotelling's original model~\cite{10.2307/2224214}, where clients simply select the nearest facility.
For $\alpha=1$, where clients are oblivious to distances, this corresponds to the client cost function in simple load balancing games on identical machines~\cite{vocking2007selfish}, where clients select the least loaded facility.

The \emph{utility} $\utilityPlayer(\strategyProfile, \clientStrategyProfile)$ of a facility  $i$ for facility location vector $\strategyProfile$ and some client mapping $\clientStrategyProfile$ equals its induced load, that is $$\utilityPlayer(\strategyProfile, \clientStrategyProfile) = \loadPlayer(\strategyProfile, \clientStrategyProfile).$$   

Similar to  (approximate) pure Nash equilibria  we define  (approximate) pure equilibria in the two-stage game  using the concept of subgame perfect equilibria.
We consider an approximate variant in which the players of the first stage (our facilities) are satisfied with approximate states while the client agents in the second stage still play optimal strategies.

\paragraph{Approximate Pure Subgame Perfect Equilibrium}
A strategy profile $(\strategyProfile, \choiceFunctions)$ is a $\rho$-\emph{approximate pure subgame perfect equilibrium} ($\rho$-SPE) if and only if  the following two conditions are satisfied:
\begin{enumerate}
	\item for all $\player \in \players$, $\utilityPlayer(\strategyProfile, \choiceFunctions) \geq \rho \cdot \utilityPlayer((\strategyProfile_{-\player},\strategyPlayer'), \choiceFunctions)$ for all $\strategyPlayer' \in \strategiesPlayer$
	\item for all $\strategyProfile \in \strategyProfiles$ and for all $\client \in \clients$, $\costsClient(\strategyProfile, \choiceFunctions) \leq \costsClient(\strategyProfile, (\choiceFunctions_{-\client}, \choiceFunction'_{\client}))$  for any alternative choice function $\choiceFunctions'\in \clientStrategyProfiles$.
\end{enumerate}
Let $\approxSpe \subseteq \strategyProfiles\times\clientStrategyProfiles$ be the set of all $\rho$-approximate subgame perfect  equilibria in the game.
For $\rho=1$, we call the state a pure subgame perfect equilibrium.

\paragraph{Client Behavior in the Subgame}
Given a facility strategy profile~$\strategyProfile$, there always exists a client equilibrium which fulfills the second condition of the equilibrium definition. This was shown in~\cite{Peters2018}, but can also easily be verified by the following potential function:
\[
\potentialFunction(\strategyProfile,\clientStrategyProfile) = (1-\alpha)\int_{0}^{1} \delta(x,\choiceFunctions(\strategyProfile, x)) dx + \alpha \sum_{i=1}^{n} \frac{(\loadPlayer(\strategyProfile,\clientStrategyProfile))^2}{2}, 
\]
where $\delta(x,\choiceFunctions(\strategyProfile, x))$ denotes the distance from $x$ to her chosen facility $\choiceFunctions(\strategyProfile, x)$ at location $\strategy_{\choiceFunctions(\strategyProfile, x)}$, i.e., $\delta(x,\choiceFunctions(\strategyProfile, x)) = |\strategy_{\choiceFunctions(\strategyProfile, x)} - x|$. 

A client equilibrium $\clientStrategyProfile$ is a measurable client mapping, i.e., for any facility $i$ there exist finitely many intervals of clients that select facility $i$.
We extend this definition to a much stronger notion of mappings in which all the clients that select some facility $i$ form a single interval of $[0,1]$, formally $|\mathcal{J}_i(\strategyProfile, \clientStrategyProfile)| = 1$ for every facility $i$.
Thus, for any fixed facility location vector $\strategyProfile$ we consider only client mappings $\clientStrategyProfile$, where the interval $[0,1]$ is partitioned into $n$ closed intervals $\mathcal{I}_1(\strategyProfile, \clientStrategyProfile),\dots,\mathcal{I}_n(\strategyProfile, \clientStrategyProfile)$ such that for all clients $\client \in \mathcal{I}_i(\strategyProfile, \clientStrategyProfile)$ we have $\clientStrategyProfile(\strategyProfile, \client) = i$.
We call such client mappings \emph{proper client mappings}.
Moreover, by re-naming facilities we can always ensure that $\strategy_1 \leq \strategy_2 \leq \dots \leq \strategy_n$ which implies that the intervals $\mathcal{I}_1(\strategyProfile, \clientStrategyProfile),\dots,\mathcal{I}_n(\strategyProfile, \clientStrategyProfile)$ are consecutive in $[0,1]$ such that $\mathcal{I}_i(\strategyProfile, \clientStrategyProfile) = [\beta_{i-1},\beta_{i}]$ with $\beta_0 = 0$ and $\beta_{n} = 1$.
A proper client mapping that is a client equilibrium is called \emph{proper client equilibrium.}
Any measurable client equilibrium can be transformed into a proper client equilibrium without changing the utilities for the facilities.
Peters et al.~\cite{Peters2018} show that such a transformation is always possible and that it results in a unique proper client equilibrium.

Therefore, we assume in the following the clients to be in the unique proper client equilibrium for any facility location vector.
This is possible since from a facility's perspective all client equilibria induce identical loads.
For a facility location vector $\strategyProfile$ we call the corresponding unique proper client equilibrium $\clientStrategyProfile_\strategyProfile$ the \emph{$\strategyProfile$-induced} client equilibrium.
Therefore, the client strategy mapping $\clientStrategyProfile_\strategyProfile$ is implicitly given and we omit it in the following definitions:
For each facility $i$ let $\mathcal{I}_i(\strategyProfile) = \mathcal{I}_i(\strategyProfile, \clientStrategyProfile_\strategyProfile)= [\beta_{i-1}, \beta_i]$ be the interval of clients using facility $i$ in this equilibrium with $\beta_0 = 0$ and $\beta_n=1$.
The load of facility $i\in \players$ is given by $\loadPlayer(\strategyProfile) = \loadPlayer(\strategyProfile, \clientStrategyProfile_\strategyProfile) = |\mathcal{I}_i(\strategyProfile)|$, the utility of facility $i \in \players$ by  $\utilityPlayer(\strategyProfile) = \utilityPlayer(\strategyProfile, \clientStrategyProfile_\strategyProfile) = \loadPlayer(\strategyProfile)$.
The costs of a client at position $\client$ are defined by $\costsClient(\strategyProfile)=\costsClient(\strategyProfile, \clientStrategyProfile_\strategyProfile)$.

\section{Analytical Results}\label{sec:analytical}
We first prove that the potential function $\potentialFunction(\strategyProfile,\clientStrategyProfile)$ suggested in Section~\ref{model} works.

\begin{lemma}\label{lem_local_pot_min_is_EQ}
	For any facility location vector $\strategyProfile$, a measurable client mapping is a client equilibrium if and only if it locally minimizes $$\potentialFunction(\strategyProfile,\clientStrategyProfile) = (1-\alpha)\int_{0}^{1} \delta(x,\choiceFunctions(\strategyProfile, x)) dx + \alpha \sum_{i=1}^{n} \frac{(\loadPlayer(\strategyProfile,\clientStrategyProfile))^2}{2}.$$
\end{lemma}

\begin{proof}
Let $\strategyProfile$ be any fixed facility location vector. We will omit the reference to $\strategyProfile$ throughout the proof. Let $\clientStrategyProfile^*$ be any measurable client mapping for $\strategyProfile$, which locally minimizes $\Phi$. We first show that if $\clientStrategyProfile^*$ is not a client equilibrium, then there is an $\varepsilon$-deviation $\clientStrategyProfile_\varepsilon$ of $\clientStrategyProfile^*$ for which $\Phi(\clientStrategyProfile_\varepsilon) < \Phi(\clientStrategyProfile^*)$. An $\varepsilon$-deviation $\clientStrategyProfile_\varepsilon$ of $\clientStrategyProfile$ differs from $\clientStrategyProfile$ only in some interval $Z$, with $|Z| = \varepsilon > 0$, such that there exists some $i \neq j$ so that for all clients $z \in Z$ we have $\clientStrategyProfile(\strategyProfile,z) = i$ and $\clientStrategyProfile_\varepsilon(\strategyProfile,z) = j$.
Suppose that $\clientStrategyProfile^*$ is not a client equilibrium. Thus, there exists an $\varepsilon$-deviation $\clientStrategyProfile_\varepsilon$ of $\clientStrategyProfile^*$. Moreover, for each client $z \in Z$ we have $C_z(\strategyProfile, \clientStrategyProfile_\varepsilon) < C_z(\strategyProfile, \clientStrategyProfile^*)$, which yields 	$$
C_z(\strategyProfile, \clientStrategyProfile^*) - C_z(\strategyProfile, \clientStrategyProfile_\varepsilon) =  (1-\alpha)(\delta(z,i) - \delta(z,j)) + \alpha(\load_{i}(\clientStrategyProfile^*) - \load_{j}(\clientStrategyProfile_\varepsilon)) > 0.
$$
	Thus, the total cost change for all clients in $Z$ is 
	
		\begin{align*}\int_Z C_z(\clientStrategyProfile^*)\text{d}z - \int_Z C_z(\clientStrategyProfile_\varepsilon) \text{d}z = (1-\alpha)\int_Z (\delta(z,\clientStrategyProfile_i)-\delta(z,j))\text{d}z
	+ \alpha\varepsilon (\load_{i}(\clientStrategyProfile^*) - \load_j(\clientStrategyProfile_\varepsilon))
	> 0.
	\end{align*}
	
		The corresponding change in potential function value $\Phi(\clientStrategyProfile^*) - \Phi(\clientStrategyProfile_\varepsilon)$ equals
		
			\begin{align*}
		&(1-\alpha) \left(\int_0^1 \delta(x,\clientStrategyProfile^*(x))\text{d}x - \int_0^1 \delta(x,\clientStrategyProfile_\varepsilon(x))\text{d}x\right) 
		+ \alpha\left(\sum_{i=1}^n \frac{(\load_i(\clientStrategyProfile^*))^2}{2} - \sum_{i=1}^n \frac{(\load_i(\clientStrategyProfile_\varepsilon))^2}{2}\right) \\
		 = & (1-\alpha)\int_Z (\delta(z,i) - \delta(z,j))\text{d}z + \alpha\left(\frac{\load_{i}(\clientStrategyProfile^*)^2}{2} + \frac{\load_{j}(\clientStrategyProfile^*)^2}{2} - \frac{\load_{i}(\clientStrategyProfile_\varepsilon)^2}{2} - \frac{\load_{j}(\clientStrategyProfile_\varepsilon)^2}{2}\right)\\
		 = &(1-\alpha)\int_Z (\delta(z,i) - \delta(p,j))\text{d}z  + \alpha\left(\frac{\ell_{i}(\clientStrategyProfile^*)^2}{2} + \frac{(\load_{j}(\clientStrategyProfile_\varepsilon)-\varepsilon)^2}{2} - \frac{(\load_{i}(\clientStrategyProfile^*)-\varepsilon)^2}{2} - \frac{\load_{j}(\clientStrategyProfile_\varepsilon)^2}{2}\right)\\
		 = & (1-\alpha)\int_Z (\delta(z,i) - \delta(p,j))\text{d}z + \alpha\left(\frac{\ell_{i}(\clientStrategyProfile^*)^2 - (\ell_{i}(\clientStrategyProfile^*)-\varepsilon)^2 }{2} + \frac{(\ell_{j}(F_\varepsilon)-\varepsilon)^2 - \ell_{j}(\clientStrategyProfile_\varepsilon)^2}{2}\right)\\
		 = & (1-\alpha)\int_Z (\delta(z,i) - \delta(p,j))\text{d}z + \alpha\left(\frac{2\varepsilon \ell_{i}(\clientStrategyProfile^*) - 1 }{2} + \frac{-2\varepsilon \ell_{j}(\clientStrategyProfile_\varepsilon) + 1}{2}\right)\\
		= &(1-\alpha)\int_Z (\delta(z,i) - \delta(p,j))\text{d}z+ \alpha\varepsilon \left(\ell_{i}(\clientStrategyProfile^*) - \ell_{j}(\clientStrategyProfile_\varepsilon)\right)\\
		 = &\int_Z C_z(\clientStrategyProfile^*)\text{d}z - \int_Z C_z(\clientStrategyProfile_\varepsilon) \text{d}z > 0,
		\end{align*}
		where the first equation is due to the fact that only distances for clients in $Z$ and only the loads of facilities $i$ and $j$ change. Thus, $\Phi(\clientStrategyProfile_\varepsilon) < \Phi(\clientStrategyProfile^*)$. 
		Hence, we have proven that every measurable client mapping which locally minimizes $\Phi$ is a client equilibrium. 	For the other direction note that the above comparison of the change in client cost and potential function value actually proves that $\Phi$ is an exact potential function. Thus, for any client equilibrium $\clientStrategyProfile^*$ and for any $\varepsilon$-deviation $\clientStrategyProfile_\varepsilon$ of $\clientStrategyProfile^*$ it follows that $\Phi(\clientStrategyProfile^*) \leq \Phi(\clientStrategyProfile_\varepsilon)$. This yields that $\clientStrategyProfile^*$ is a local minimum of $\Phi$. 
\end{proof}
\noindent With Lemma~\ref{lem_local_pot_min_is_EQ} we can easily establish that for every facility location vector $\strategyProfile$ there exists a client equilibrium.

We analyze $\rho$-SPE for several settings. Our main goal is to show that the equilibria found by Eaton \& Lipsey~\cite{Eaton1975} for $\alpha = 0$ are also good approximate equilibria for  $\alpha \in [0,1]$ as well, i.e., $\rho$ is  small.

As shown in~\cite{Peters2018}, it holds for a ($\rho$-)SPE that for any two neighboring intervals $\mathcal{I}_i(\strategyProfile) = [\beta_{i-1},\beta_i]$, $\mathcal{I}_{i+1}(\strategyProfile) = [\beta_i,\beta_{i+1}]$  the clients at $\beta_i$ are indifferent between choosing facility $i$ or ${i+1}$ as costs are equal for both strategies: \[(1-\alpha) \cdot |\strategy_{i} - \beta_i| + \alpha \cdot \load_{i}(\strategyProfile) = (1-\alpha) \cdot |\strategy_{i+1} - \beta_i| + \alpha \cdot \load_{i+1}(\strategyProfile).\]
Taking these equations for all $n-1$ interval borders results in a system of equations, which allows us to compute the interval borders. 
In our analytical computations we make also use of the result of~\cite{Peters2018} that the best response of the both external facilities $1$ and $n$ is to locate at $\beta_1$ and $\beta_{n-1}$, respectively. Furthermore it holds that the best response $s_\player^{\text{best}}$ of a facility $\player$ is inside her corresponding interval, i.e., $s_\player^{\text{best}} \in \mathcal{I}_i(\strategyProfile, \choiceFunctions)$.
If a facility~$\player$ can improve by changing her strategy from $s_{\player}$ to another strategy $s_{\player}'$, we denote her improvement factor as $\rho_{s_i}'$ and the new interval border as $\beta_i'$. Both $\rho_{s_i}'$ and $\beta_i'$ depend on $(s_{-i},s_i')$ but we will omit the reference to $(s_{-i},s_i')$  since it will be clear from the context. 

\subsection{Three Facilities}
We start with three facilities and show that one facility at $\frac{1}{2}$ and the the other two equidistant to the left and right, respectively, with a suitably chosen gap yields a good $\rho$-SPE. 

\begin{theorem}\label{thm:3facilities_alpha}
	For $n = 3$ the game has a $\rho$-SPE with $\rho = \frac{1-\alpha^2+\sqrt{17+\alpha(16+2\alpha+\alpha^2)}}{4-2(\alpha-2)\alpha}$.
\end{theorem}

\begin{proof}
	Consider $s = (s_1, \frac{1}{2}, 1-s_1)$. 
	he clients' interval splits at  $\beta_1$ and $\beta_2$ with 
	\begin{eqnarray}
	(1-\alpha)(\beta_1-s_1)+\alpha\beta_1 &=& (1-\alpha)(\frac{1}{2}-\beta_1)+\alpha(\beta_2 - \beta_1), \\
	\beta_2 & = &  1-\beta_1.
	\end{eqnarray}
	Therefore, \begin{eqnarray}
	\beta_1 &=& \frac{1+\alpha+2s_1-2\alpha s_1}{4+2\alpha}, \\
	\beta_2 & = &  \frac{3+\alpha-2s_1+2\alpha s_1}{4+2\alpha}.
	\end{eqnarray}	
Since player $F_1$ and $F_3$ are equivalent we only consider player $F_1$.
The best response of $F_1$ is to locate at $\beta'_1$. So it follows from
\begin{eqnarray}
\alpha\beta'_1 &=& (1-\alpha)(\frac{1}{2}-\beta'_1)+\alpha(\beta'_2 - \beta'_1), \\
(1-\alpha)(\beta'_2-\frac{1}{2})+\alpha(\beta'_2-\beta'_1) &=& (1-\alpha)|(1-s_1)-\beta'_2|+\alpha(1-\beta'_2),
\end{eqnarray}
that 
\begin{eqnarray}
\beta'_1 &=& \frac{2+\alpha-2\alpha s_1+\alpha^2(-1+2s_1)}{4+4\alpha-2\alpha^2}, \\
\beta'_2 & = &  \frac{3+3\alpha+2\alpha^2(-1+s_1)-2s_1}{4+4\alpha-2\alpha^2}.
\end{eqnarray}
\noindent Thus, facility $F_1$ can improve by a factor of $\rho_1 = \frac{2(2+\alpha)(2+\alpha-2\alpha s_1+\alpha^2(-1+2s_1))}{(4+4\alpha-2\alpha^2)(1+\alpha+2s_1-2\alpha s_1)}$ (as well as $F_3$, respectively).
	By our choice of~$s_1$, we will ensure that $s_1' < \frac12$ is not a best response.

	We now consider facility $F_2$. As $s$ is symmetric, we can assume, without loss of generality, that the best response of facility $2$ is a position $s_2' < \frac12$. 
	For her best response $s'_2$, we consider two cases:
	
	\begin{itemize}
		\item $s_2' \leq s_1$: In this case, the utility of facility $F_2$ is equal to the length of the first interval ending at point $\beta_1$. So, as discussed for player $F_1$, the best response is $s_2' = \beta_1'$. Hence, $s_2' = \frac{\alpha + 2s_1-2\alpha s_1}{2+2\alpha-\alpha^2}$ and facility $F_2$ can improve by $\rho_2 = \frac{(2+\alpha)(\alpha+2s_1-2\alpha s_1)}{(2+2\alpha-\alpha^2)(1+2(-1+\alpha)s_1)}$.
		\item $s_2' > s_1$: Note that $s_2' < \frac{1}{2}$ is symmetric to $1-s_2'$. So we have
		\begin{align}
		(1-\alpha)|\beta'_1-s_1|+\alpha\beta'_1 &= (1-\alpha)(\frac{1}{2}-\epsilon-\beta'_1)+\alpha(\beta'_2 - \beta'_1), \\
		\hspace*{-0.75cm}(1-\alpha)(\beta'_2-(\frac{1}{2}-\epsilon))+\alpha(\beta'_2-\beta'_1) &= (1-\alpha)((1-s_1)-\beta'_2)+\alpha(1-\beta'_2).
		\end{align}

		The utility $u_2 = \beta_2' - \beta_1' = \frac{1-2\epsilon+\alpha^2(1-2\epsilon-2s_1)-2s_1+\alpha(-3+4\epsilon+4s_1)}{(-4+\alpha)\alpha}$ becomes larger the greater $\epsilon > 0$ gets. In particular it is better for player $F_2$ to be at the same location as player $F_1$ than to be between player $F_1$ and $F_3$.
	\end{itemize}
Choosing $s_1 = \frac{-3+(\alpha-4)\alpha+\sqrt{17+\alpha(16+2\alpha+\alpha^3)}}{4(a-1)^2}$ minimizes the maximum of $\rho_1$ and $\rho_2$ and both evaluate, for $0 < \alpha < 1$, to 
\[\rho = \frac{1-\alpha^2+\sqrt{17+\alpha(16+2\alpha+\alpha^3)}}{4-2(-2+a)a}.\qedhere\]

\end{proof}
\noindent Theorem~\ref{thm:3facilities_alpha} yields directly the following statement.

\begin{corollary}\label{corll:3facilities}
	For $\alpha = 0$ and $n = 3$ the game has a $\rho$-SPE with $\rho = \frac{1}{4}(1+\sqrt{17})$.
\end{corollary}

\noindent We now show that Corollary~\ref{corll:3facilities} is tight as it yields the $\rho$-SPE with minimal $\rho$ for Hotelling's original model.
\begin{theorem}\label{thm:3(b)}
	For $\alpha = 0$ and $n = 3$ the game does not have $\rho$-SPE with $\rho < \frac{1}{4}(1+\sqrt{17})$.
\end{theorem}

\begin{proof}
	We need to consider three cases: all facilities in the same location, two choosing the same location, and all three choosing different locations.\\
	\begin{itemize}
	\item[Case 1:] Consider all facilities choosing the same location, hence, $s = (s_1, s_1, s_1)$ and $\loadPlayer(\strategyProfile) = \frac{1}{3}$. Each player is equivalent, so we only consider facility $F_1$. Without loss of generality we can assume $s_1 \leq \frac{1}{2}$. The best response for a facility $i$ is to move to $s_i' = s_1 + \epsilon$ for some $\epsilon >0$, which results in an approximation factor $$\rho'_i = \lim_{\epsilon \rightarrow 0} \frac{(1-s_i')-\frac{\epsilon}{2}}{\loadPlayer(\strategyProfile)} \geq \frac{3}{2}.$$
	\item[Case 2:]  Consider two facilities choosing the same location, hence, $s = (s_1, s_2, s_2)$. It holds that $s_1 < \frac{1}{2} \leq s_2$, as otherwise there would be a facility $i$ with $\rho_{s_i} \geq 2$. The best response for facility  $1$ is $s_1' = s_2 - \epsilon$ some $\epsilon >0$, which leads to $$\rho'_{s_1} = \lim_{\epsilon \rightarrow 0} \frac{s_2-\frac{\epsilon}{2}}{\frac{s_1+s_2}{2}} = \frac{2s_2}{s_1+s_2}.$$ Since $\rho_1 < \frac{1}{4}(1+\sqrt{17})$ it follows
	\begin{eqnarray}
	\frac{7s_2-\sqrt{17}s_2}{1+\sqrt{17}} \leq s_1 \leq  \frac{1}{2} \label{eq_s1}  \leq  s_2 \leq \frac{-1-\sqrt{17}}{-14+2\sqrt{17}}.
	\end{eqnarray}
	
	Facility~$2$ and $3$ are equivalent. A possible strategy change for facility $2$ is either $s_2' = s_1 - \epsilon$ which results in  $$\rho'_{s_2} = \lim_{\epsilon \rightarrow 0} \frac{s_1 - \frac{\epsilon}{2}}{\frac{2-s_1-s_2}{4}} =  \frac{4s_1}{2-s_1-s_2}$$ and therefore
	\begin{eqnarray}
	\frac{7s_2-\sqrt{17}s_2}{1+\sqrt{17}} &\leq s_1 \leq & \frac{2+\sqrt{17}-s_2-\sqrt{17}s_2}{17+\sqrt{17}} \label{f8}, \\
	\frac{1}{2} &\leq s_2 \leq & \frac{-1-\sqrt{17}}{-30+2\sqrt{17}} \label{f9}
	\end{eqnarray}
	or $s_2'' = s_2 + \epsilon$ which results in $$\rho''_{s_2} = \lim_{\epsilon \rightarrow 0} \frac{1-(s_2+\frac{\epsilon}{2})}{\frac{2-s_1-s_2}{4}} =  \frac{4(1-s_2)}{2-s_1-s_2}.$$ However, $\rho''_{s_2} < \frac{1}{4}(1+\sqrt{17})$ contradicts \eqref{f8} and \eqref{f9}.So there is no valid choice of $s_1$ and $s_2$ such that $\rho'_{s_1}$, $\rho'_{s_2}$ and $\rho''_{s_2}$ are smaller than $\frac{1}{4}(1+\sqrt{17})$.\\
	\item[Case 3:] 
	Consider all facilities choosing different locations, hence, $s = (s_1, s_2, s_3)$. It holds that $s_1 < \frac{1}{2} \leq s_2 < s_3$, since otherwise $\rho_{s_i} \geq 2$. Like in the previous case the best response for facility $1$ is $s_1' = s_2 - \epsilon$ which leads to $\rho'_{s_1} = \frac{2s_2}{s_1+s_2}$. 
	A possible strategy change for facility $2$ is $s_2' = s_1 - \epsilon$ with  $$\rho'_{s_2} = \lim_{\epsilon \rightarrow 0} \frac{s_1 - \frac{\epsilon}{2}}{\frac{s_3-s_1}{2}} = \frac{2s_1}{s_3-s_1}$$ and therefore 
		\begin{eqnarray}
	 s_1 \leq  \frac{s_3+\sqrt{17}s_3}{9+\sqrt{17}}.
	\end{eqnarray}
	or $s_2'' = s_3 + \epsilon$ which leads to $$\rho''_{s_2} = \lim_{\epsilon \rightarrow 0} = \frac{1-(s_3-\frac{\epsilon}{2})}{\frac{s_3-s_1}{2}} = \frac{2s_1+s_3}{s_3-s_1}.$$ Hence, it has to hold 
	\begin{eqnarray}
	& s_1 \leq  \frac{-8+9s_3+\sqrt{17}s_3}{1+\sqrt{17}} &\text{ and } \frac{8}{9+\sqrt{17}} < s_3 \leq \frac{9+\sqrt{17}}{10+2\sqrt{17}}   \label{eq_f6}, \\
	\text{or } & s_1 \leq  \frac{s_3+\sqrt{17}s_3}{9+\sqrt{17}}  &\text{ and }  \frac{9+\sqrt{17}}{10+2\sqrt{17}} < s_3. \label{eq_f7}
	\end{eqnarray}
	
	Facility $3$ has the possibility to move to $s_3' = s_2 + \epsilon$ with $$\rho'_{s_3} = \lim_{\epsilon \rightarrow 0} \frac{1-(s_2+\frac{\epsilon}{2})}{\frac{2-s_2-s_3}{2}} = \frac{2(1-s_2)}{2-s_2-s_3},$$ so 
	\begin{eqnarray}
	& s_3 \leq  \frac{-5+3\sqrt{17}}{2+2\sqrt{17}}   \label{eq_f8}, \\
	\text{or } & s_2 \geq  \frac{-6+2\sqrt{17}-s_3-\sqrt{17}s_3}{-7+\sqrt{17}} \text{ and }  s_3 >  \frac{-5+3\sqrt{17}}{2+2\sqrt{17}} . \label{eq_f9} 
	\end{eqnarray}
	or $s_3'' =  s_1 - \epsilon$ with $$\rho''_{s_3} = \lim_{\epsilon \rightarrow 0} \frac{s_1 -\frac{\epsilon}{2}}{\frac{2-s_2-s_3}{2}} = \frac{2s_1}{2-s_2-s_3}.$$ However, $\rho''_{s_3} < \frac{1}{4}(1+\sqrt{17})$ contradicts \eqref{eq_f8} and \eqref{eq_f9}. So there is no valid solution with $\rho'_{s_1}$, $\rho'_{s_2}, \rho''_{s_2}$, $\rho'_{s_3}$ and $\rho''_{s_3}$ smaller than $\frac{1}{4}(1+\sqrt{17})$. \qedhere
	\end{itemize}
\end{proof}

\subsection{Uniformly Distributed Facilities}
As a warm-up, we consider the uniform distribution $\strategyProfile_\text{opt}$ of all facilities on the line, which is defined as $\strategyProfile_\text{opt} = (\strategy_1, \ldots, \strategy_n)$ with $\strategy_i = \frac{2i-1}{2n}$ for $i \in \{1, \ldots, n\}$. See Figure~\ref{fig:opt_positions} for an illustration. 
Note, that this facility placement minimizes the average client cost.

\begin{figure}[ht]
	\centering
	\includegraphics[width=13cm]{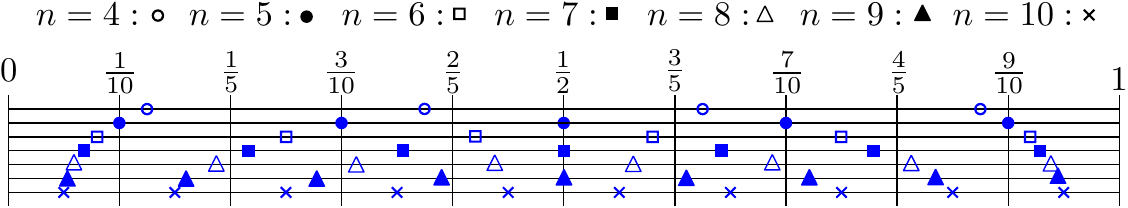}
	\caption{Facility positions in $\strategyProfile_\text{opt}$ for $4\leq n \leq 10$.}
	\label{fig:opt_positions}
\end{figure}

\noindent For a small number of players, i.e., $4\le n \le 10$, we determine $\rho$ explicitly as a function of~$\alpha$. 

\begin{theorem} The locations $\strategyProfile_\text{opt}$ yields  a $\rho_n$-SPE in the game with $n$ facilities with the following values of $\rho_n$.\\
%	\noindent
	{ $\rho_4 = \frac{1}{2} + \frac{2(\alpha^2-2)}{(\alpha-1)\alpha(4+\alpha)-4}$, \\ $\rho_5 = \frac{12+\alpha(4+\alpha(\alpha(\alpha-2)-10))}{8+\alpha(2+\alpha)(4+(\alpha-6)\alpha)}$,\\
		$\rho_6 = \frac{1}{2} + \frac{16-16\alpha^2+3\alpha^4}{16+\alpha(16+\alpha(\alpha(\alpha(5+\alpha)-12)-20))}$,\\
		$\rho_7 = \frac{\alpha(\alpha(64+\alpha(16+\alpha(\alpha(\alpha-3)-21)))-16)-48}{\alpha(\alpha(48+\alpha(32+\alpha((\alpha-6)\alpha-18))))-32)-32}$,\\
		$\rho_8 = \frac{1}{2} + \frac{4(\alpha^2-2)(8-8\alpha^2+\alpha^4)}{a((\alpha-2)\alpha(2+\alpha)(\alpha(\alpha(7+\alpha)-20)-28)-64)-64}$,\\
		$\rho_9 = \frac{192+\alpha(64+\alpha(\alpha(4+\alpha)(\alpha(56+\alpha((\alpha-8)\alpha-4))-24)-352))}{128+(\alpha-2)\alpha(\alpha(2+\alpha)(48+\alpha(48+\alpha((\alpha-8)\alpha-28)))-64)}$,\\
		$\rho_{10} = \frac{1}{2}+ \frac{256-512\alpha^2+336\alpha^4-80\alpha^6+5\alpha^8}{256+\alpha(256+\alpha(\alpha(\alpha(432+\alpha(240+\alpha(\alpha(9\alpha+\alpha^2-40)-120))))-448)-576))}$}.
\end{theorem}

\begin{proof}
We compute the interval borders $\beta_1, \ldots, \beta_{n-1}$ for the proper client equilibrium by solving the system with $n-1$ equations\begin{eqnarray}
(1-\alpha)|s_i-\beta_i| + \alpha \load_i(s) = (1-\alpha)|s_{i+1}-\beta_i| + \alpha \load_{i+1}(s)
\end{eqnarray}
for $i \in \{1, \ldots, n-1\}$. A facility $F_i$ obtains a load of $\load_i(\strategyProfile_\text{opt}) = \frac1n$ in the strategy vector~$\strategyProfile_\text{opt}$. Therefore the strategy changes $s_i' < s_1$ and $s_i' > s_n$ is not an improvement for an arbitrary facility $i$ since $\beta_1 > s_1$ and $\beta_{n-1} < s_n$. So the best response for a facility $i$ is to locate inside the interval $[s_1,s_n]$. To compute the best response of a facility $i$ we have to check all possible strategy changes. So $i$ can be located in each of the subintervals$[s_1,s_2],[s_2,s_3], \ldots, [s_{n-1}, s_n]$. Solving \begin{eqnarray*}
(1-\alpha)(\beta_1' - s_1')| + \alpha \beta_1' &=& (1-\alpha)|s_{2}-\beta_1'| + \alpha (\beta_2' - \beta_1'), \\
(1-\alpha)(\beta_2' - s_2)| + \alpha (\beta_2'-\beta_1') &=& (1-\alpha)|s_{3}-\beta_2'| + \alpha (\beta_3' - \beta_2'), \\
\ldots \\
(1-\alpha)(\beta_{n-1}' - s_{n-1})| + \alpha (\beta_{n-1}'-\beta_{n-2}') &=& (1-\alpha)|s_{n}-\beta_{n-1}'| + \alpha (1 - \beta_{n-1}'), 
\end{eqnarray*}
yields the new interval borders $\beta_i'$ for $1 \leq i \leq n-1$ when facility $1$ changes her strategy to $s_1' \in [s_1, s_2]$. Together with the result that the best response of facility $1$ is to locate at her interval border $\beta_1'$, we can easily calculate the approximation factor $\rho$ for this case. 

To check how good the other strategy changes are, we have to construct a modified system of equations, where we respect that the considered facility $i$ is not anymore in the consecutive order $s_1 \leq s_2 \leq \ldots \leq s_n$ in $[0,1]$. By setting up a suitable system of equations for each case $s_i' \in [s_k,s_{k+1}]$ for $k \in \{1, \ldots, n-1\}$ we address the problem.

So we can verify for all facilities $i$ for $i \in \{1, \ldots, n\}$ all possible strategy changes with the help of the system of equations. It turns out that for all $n \leq 10$ the facilities $1$ and $n$ have the highest possible improvement by moving to the new interval border $\beta_1'$ and $\beta_{n-1}'$, respectively.
\end{proof}

\noindent Based on the results of the previous section and our agent-based simulations (cf.\ Section~\ref{subsec:empirical_support_conjecture}) we derive the following conjecture for an arbitrary number of facilities.

\begin{conjecture}\label{conj:optimum_highest_improvement}
	Given a game with $n > 3$ facilities and the state $\strategyProfile_\text{opt} = (\strategy_1, \ldots, \strategy_n)$ with $\strategy_i = \frac{2i-1}{2n}$ for all $i \in \{1, \ldots, n\}$. Then one of
	the outmost facilities, $1$ or $n$, has the highest possible improvement factor by changing her strategy towards the middle to the new interval border $s_1' = \beta'_1$ or $s_n' = \beta'_{n-1}$.
\end{conjecture}

\noindent Using generalized continued fractions, define
\begin{equation*}
\contFraQua^m := \contFracOpe_{j=1}^{m}{\frac{-\alpha^2/4}{1}} =  \cfrac{-\alpha^2/4}{1 +\cfrac{-\alpha^2/4}{1 + \ddots \cfrac{-\alpha^2/4} {1 }} } \text{\quad and}
\end{equation*}

\begin{align*}
\approxOptimumGeneral = 
\frac{n}{1+\frac{2}{1+\alpha} \contFraQua^{n-2}}
 \Bigg(
\frac{1-\alpha}{1+\alpha}\frac{3}{2n}
+\sum_{k=2}^{n-1}{\frac{1-\alpha}{1+\alpha} \frac{2k}{n} \prod_{j=n-k}^{n-2}{\left(-\frac{2}{\alpha} \contFraQua^{j}\right)} }
+\frac{\alpha}{1+\alpha}\prod_{j=1}^{n-2}{\left(-\frac{2}{\alpha} \contFraQua^{j}\right)}
\Bigg)
.\end{align*}

\noindent Using Conjecture~\ref{conj:optimum_highest_improvement} and the definition of $\approxOptimumGeneral$ we can state the following approximation guarantee for arbitrary $n$.

\begin{theorem}\label{theorem:optimum_arbitrary_n}
	Assume Conjecture~\ref{conj:optimum_highest_improvement} holds for $n>3$ facilities. Then the game has a $\rho$-SPE with $\rho = \approxOptimumGeneral$.
\end{theorem}

\begin{proof}
	Consider the state $\strategyProfile_\text{opt} = (\strategy_1, \ldots, \strategy_n)$ with $\strategy_i = \frac{2i-1}{2n}$ for $ i \in \{1, \ldots, n\}$.
	The clients' intervals split at $\beta_i= \frac{i}{n}$ for $ i \in \{1, \ldots, n-1\}$, so each facility $i$ has a utility of $\utilityPlayer(\strategyProfile) = \frac{1}{n}$.
	Using Conjecture~\ref{conj:optimum_highest_improvement}, we only need to consider facility $1$ with a move to her new interval border $\beta'_1$, formally $\strategyProfile'= (\strategyProfile_{-1}, \beta'_1)$ and we formalize the new state with a system of linear equations.
	
	\begin{eqnarray*}
	\nal(\beta'_1-\strategy'_1)\ +\ \al(\beta'_1-\beta'_0) & =& \nal(\strategy_2-\beta'_1)\ +\ \al(\beta'_2-\beta'_1),\\
	\nal(\beta'_i-\strategy_i)\ +\ \al(\beta'_i-\beta'_{i-1}) & =& \nal(\strategy_{i+1}-\beta'_i)\ +\ \al(\beta'_{i+1}-\beta'_i)\\
		& & \forall i \in \{2, \ldots , n-1 \}.
	\end{eqnarray*}
	We solve this system for $\beta'_1$ using Gaussian elimination and generalized continued fractions.
	Since we consider the first facility, we have $\utility_1(\strategyProfile') = \beta'_1$. The derivation is similar to the proof of Theorem~\ref{theorem:arbitrary_n_double}.
	Together with $\utility_1(\strategyProfile)=\frac{1}{n}$ and Conjecture~\ref{conj:optimum_highest_improvement} we get $\rho = \frac{\utility_1(\strategyProfile')}{\utility_1(\strategyProfile)} = n\beta'_1$
	which equals $\approxOptimumGeneral$ by definition.
\end{proof}

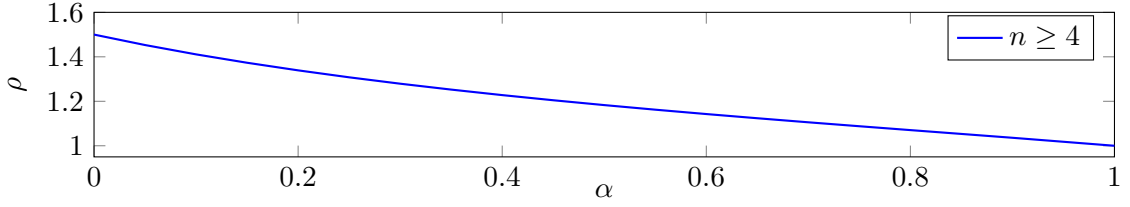
\begin{figure}[htb]
	\centering
	\begin{tikzpicture}
	\begin{axis}[xmin=0, xmax=1, ymin=0.95, ymax=1.6, samples at={0,0.05,...,0.95,1}, width=\columnwidth, height=3.5cm, xlabel={$\alpha$}, ylabel={$\rho$}, ylabel near ticks, xtick={0,0.2,0.4,0.6,0.8,1},  x label style={ below=-3mm}]
	\addplot[blue, thick] {(x^3+7*x^2-4*x-12)/(2*(x^3+3*x^2-4*x-4))};
	\legend{$n \geq 4$}
	\end{axis}
	\end{tikzpicture}
	\vspace*{-1em} % use for manual adjustment at the end
	\caption{Approximation factor $\rho$ for $\strategyProfile_\text{opt}$ as a function of $\alpha$.}
	\label{fig:facilityApproxFactorsDistributedFacilities}
\end{figure}

\noindent The influence of the number of facilities $n$ is negligible in $\approxOptimumGeneral$, so Figure~\ref{fig:facilityApproxFactorsDistributedFacilities} shows the approximation factor as a function of $\alpha$.
For large values of $\alpha$ the factor is close to $1$, which is to be expected as the actual location of the facilities are less important. However for the remaining range of $\alpha$, facilities can improve significantly.

\subsection{Co-locating Facilities}
We study a facility placement $\strategyProfile_\text{pair}$ which
was proposed by Eaton~\& Lipsey~\cite{Eaton1975} and respects the principle of minimum differentiation since it  consists of co-located pairs of facilities. 
We show for  $n \leq 10$ that the placements $\strategyProfile_\text{pair}$ yield $\rho$-SPE for all $\alpha$ with surprisingly small values of $\rho$. 
For an even number of players $n = 2k$ the placement 
is $\strategyProfile_\text{pair} = (s_1, \ldots, s_n)$ and for an odd number of players $n = 2k-1$ the placement is $\strategyProfile_\text{pair} = (s_1, \ldots,s_{k-1}, s_k, s_{k+2}, \ldots, s_{n+1})$ with $s_{2 i -1} = s_{2i } = \frac{2i -1}{2k}$ for $i \in \{ 1, \ldots, k\}$ for some $k \in \mathbb{N}$ (see Figure~\ref{fig:eaton_lipsey_positions}). Eaton \& Lipsey~\cite{Eaton1975} proved that $\strategyProfile_\text{pair}$ is a SPE for $\alpha = 0$. Moreover, it trivially is also a SPE for $\alpha = 1$ since any facility placement is a SPE for $\alpha=1$.
\begin{figure}[ht]
	\centering
	\includegraphics[width=13cm]{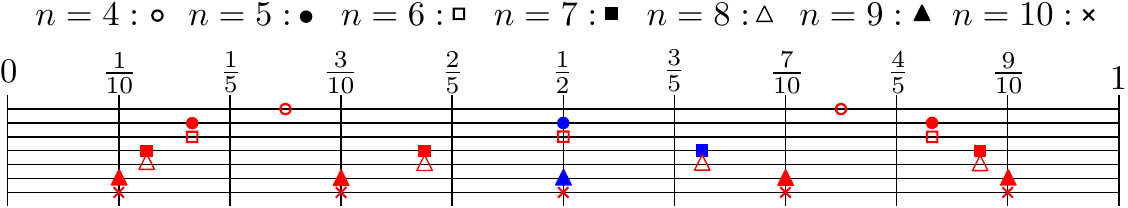}
	\caption{Facility placements $\strategyProfile_\text{pair}$ for $4\leq n \leq 10$. Co-located facilities are colored red, single facilities are colored blue.}
	\label{fig:eaton_lipsey_positions}
\end{figure}

\begin{theorem}\label{thm:eaton_lipsey}
	The locations $\strategyProfile_\text{pair}$ yields  a $\rho_n$-SPE in the game with~$n$ facilities with the following values of $\rho_n$.\\
%	\noindent
	{ $\rho_4 = \frac{4 + \alpha - \alpha^2}{4}$, \\
		$\rho_5 = \frac{(4+\alpha)(\alpha(\alpha(3+\alpha)-3)-4)}{(2+\alpha)(\alpha(5\alpha-2)-8)}$, \\
		$\rho_6 = \frac{\alpha(4-\alpha(\alpha-7))-16}{2(\alpha(4+\alpha)-8)}$,\\
		$\rho_7 = \frac{(64-64\alpha+7\alpha^3)(16+\alpha(2+\alpha)(\alpha(\alpha-3)-2))}{2(32+\alpha(\alpha(\alpha-10)-16))(16+\alpha(\alpha-16+\alpha^2))}$, \\
		$\rho_8 = \frac{64-\alpha(48+\alpha(24+(\alpha-17)\alpha))}{4(16+\alpha(\alpha-16+\alpha^2))}$, \\
		$\rho_9 = \frac{(32+(\alpha-4)\alpha(2+\alpha)(1+2\alpha))(\alpha(4+3\alpha)-16)}{(\alpha-2)(4+\alpha)(64+\alpha(\alpha(\alpha-24)-32))}$,\\
		$\rho_{10} = \frac{\alpha(320-\alpha(\alpha(120(\alpha-31)\alpha)-16))-256}{2(\alpha(\alpha-4)(\alpha(4+3\alpha)-48)-128)}$}.
\end{theorem}

\begin{proof}
	We compute the interval borders $\beta_1, \ldots, \beta_{n-1}$ for the proper client equilibrium by solving the system with $n-1$ equations  
	\begin{eqnarray*}(1-\alpha) |\strategy_{i} - \beta_i| + \alpha  \load_{i}(\strategyProfile) = (1-\alpha)  |\strategy_{i+1} - \beta_i| + \alpha  \load_{i+1}(\strategyProfile) \text{ for } i \in \{1, \ldots, n-1\}. \end{eqnarray*}
	The strategy change $s_i' < s_1$ and $s_i' > s_n$ is not an improvement since $\beta_1 \geq s_1$ and $\beta_{n-1} \leq s_n$ for an arbitrary facility $i$. As already mentioned, the best response for the leftmost and rightmost facility is to locate at $\beta_1$ and $\beta_{n-1}$, respectively. Together with $\load_i(\strategyProfile)$, it can be checked that this is not an improvement for facility $i$.
	So the best response for facility $\player$ is to locate inside the interval $[s_1,s_n]$.
	To compute the best response of facility $i$ we have to check all possible strategy changes. So $i$ can be located in each of the subintervals $[s_1,s_2]$, $[s_2,s_3]$, $\ldots$, $[s_{n-1},s_n]$. 
	Since facility $1$ and  $2$ are equivalent, we just have to consider facility $2$ and her strategy changes. Solving
	\begin{eqnarray*}
	(1-\alpha) |\beta_1' - s_1| + \alpha  \beta_1' &=&  (1-\alpha) (s_2'-\beta_1')+\alpha  (\beta_2'-\beta_1'), \\
	(1-\alpha) (\beta_2' - s_2') + \alpha  (\beta_2'-\beta_1') &=&  (1-\alpha) |s_3-\beta_2'|+\alpha  (\beta_3'-\beta_2'), \\
	 \ldots \\
	(1-\alpha) |\beta_{n-1}' - s_{n-1}| + \alpha  (\beta_{n-1}'-\beta_{n-2}') &=&   
	(1-\alpha) |s_n-\beta_{n-1}'|+\alpha (1-\beta_{n-1}'), 
	\end{eqnarray*}
	\noindent yields the new interval borders $\beta_i'$ for $1\leq i \leq n-1$ when facility~$2$ changes her strategy to $s_2' \in [s_2,s_3]$. Together with the result that the best response $s_\player^{\text{best}}$ of a facility $i$ is inside her corresponding interval, i.e., $s_\player^{\text{best}} \in \mathcal{I}_i(\strategyProfile, \choiceFunctions)$, so $s_2' \in [\beta_1',\beta_2']$ and it can be checked that $s_2' = \beta_2'$ is the best response for $s_2' \in [s_2,s_3]$.
	
	To check how good the other strategy changes are, we have to construct a modified system of equations, which respects that the considered facility $i$ is not anymore in the consecutive order $\strategy_1 \leq \strategy_2 \leq \dots \leq \strategy_n$ in $[0,1]$. This is done by setting up a suitable system of equations for each case $s_i' \in [s_k,s_{k+1}]$ for $1 \leq k \leq n-1$.
	So we can verify for all facilities $i$ for $1\leq i \leq n$ all possible strategy changes with the help of the system of equations.
	It turns out that for all $n \leq 10$ the facilities~$1$ and $2$, respectively have the highest possible improvement by moving to the new interval border $\beta_2'$.
\end{proof}

\noindent Our analytical results and agent-based simulations (see Section~\ref{subsec:empirical_support_conjecture}) suggest that the outmost facilities $1$ and $2$, respectively, yield the highest possible improvement by moving to the new interval border $\beta_2'$. Therefore we state the following conjecture.

\begin{conjecture}\label{conj:double_highest_improvement}
	Given a game with $n > 3 $ facilities and the state $s = (s_1, \ldots, s_n)$ for $n = 2k$ and $s = (s_1, \ldots,s_{k-1}, s_k, s_{k+2}, \ldots, s_{n+1})$ for $n = 2k-1$ for some $k \in \mathbb{N}$ with $s_{2 i -1} = s_{2i } = \frac{2i -1}{2k}$ for $i \in \{1, \ldots, k\}$. Then one of the leftmost facilities, $1$ or $2$, has the highest possible improvement factor by changing her strategy towards the middle to the new interval border $\beta_2'$.
\end{conjecture}

\noindent
With the help of generalized continued fractions, define
\begin{equation*}
\contFraLin^m := \contFracOpe_{j=1}^{m}{\frac{-\alpha/4}{1}} = \cfrac{-\alpha/4}{1 +\cfrac{-\alpha/4}{1 + \ddots \cfrac{-\alpha/4} {1 }} },
\end{equation*}

\begin{eqnarray*}
\beta'_1 &=& 
\frac{1}{1-\frac{\alpha}{2\left(\alpha + 1 + 2 \contFraLin^{n-3} \right)}}
\Bigg(
\frac{1-\alpha}{2n}
+ \frac{1-\alpha}{2(\alpha +1)} \frac{1}{1+\frac{2}{\alpha+1}\contFraLin^{n-3}} \frac{3}{n}
- \frac{2}{\alpha+1}\frac{1}{1+\frac{2}{\alpha+1}\contFraLin^{n-3}}\\ && \Bigg(\sum_{k=2}^{n/2-1}\Bigg( \frac{(-2)^{2k-3}}{a^{k-1}}\frac{2(1-\alpha)k}{n}
\prod_{j=n-2k}^{n-3}{\contFraLin^{j} } \Bigg) + \frac{(-2)^{n-4}}{2a^{(n-4)/2}}   \prod_{j=1}^{n-3}{\contFraLin^{j} }  \Bigg)
\Bigg),
\end{eqnarray*}

\begin{eqnarray*}
\beta'_2 & =& \frac{1}{1+\frac{2}{1+\alpha}\contFraLin^{n-3}}
\Bigg(
\frac{\alpha}{1+\alpha} \beta'_1
+\frac{1-\alpha}{1+\alpha}\frac{3}{n}
+ \sum_{k=2}^{n/2-1} 
\Bigg(\frac{-4}{1+\alpha} \frac{(-2)^{2k-3}}{a^{k-1}} \frac{2(1-\alpha)k}{n}
\prod_{j=n-2k}^{n-3}{\contFraLin^{j} } 
\Bigg)	\\
&&- \frac{2}{1+\alpha}\frac{(-2)^{n-4}}{a^{(n-4)/2}} \prod_{j=1}^{n-3}{\contFraLin^{j} }
\Bigg)  \text{\quad and \quad} \approxPairingGeneral = n \left(\beta'_2-\beta'_1\right).
\end{eqnarray*}
\noindent Using Conjecture~\ref{conj:double_highest_improvement} and $\approxPairingGeneral$ we can state the approximation factor for an arbitrary even number of facilities and an arbitrary $\alpha$.

\begin{theorem}\label{theorem:arbitrary_n_double}
	Assuming  Conjecture \ref{conj:double_highest_improvement} holds, for $n>3$ facilities with $n=2k$ and $k\in \naturals$, the game has a $\rho$-approximate pure subgame perfect equilibrium with $\rho =\approxPairingGeneral$.
\end{theorem}

\begin{proof}
	Consider the state $\strategyProfile = (\strategy_1, \ldots, \strategy_n)$ with $\strategy_{2i-1}=\strategy_{2i} = \frac{2i-1}{2k}$ for $i \in \{1, \ldots, k\}$.
	The clients' intervals split at $\beta_i= \frac{i}{2k}\ \forall i \in \{1, \ldots, 2k-1\}$, so each facility $i$ has a utility of $\utilityPlayer(\strategyProfile) = \frac{1}{n}$.
	Using Conjecture~\ref{conj:double_highest_improvement}, we only consider facility $2$ with a move to her new interval border $\beta'_2$.
	The following system of linear equations characterizes the new state $\strategyProfile' = (\strategyProfile_{-2}, \beta'_2)$:
	{
		\begin{eqnarray*}
		\nal(\beta'_1-\strategy_1) + \al(\beta'_1-\beta'_0) & =& \nal(\strategy'_2-\beta'_1) + \al(\beta'_2-\beta'_1)\\
		\nal(\beta'_2-\strategy'_2) + \al(\beta'_2-\beta'_1) & =& \nal(\strategy_3-\beta'_2) + \al(\beta'_3-\beta'_2) \\
		\nal(\beta'_{2i-1}-\strategy_{2i-1})\ +\ \al(\beta'_{2i-1}-\beta'_{2i-2}) & = & \nal(\beta'_{2i-1}-\strategy_{2i}) + \al(\beta'_{2i}-\beta'_{2i-1}) \\
		\forall i\ \in\ \{2,\ \ldots\ ,\ k-1\} \\
			\nal(\beta'_{2i}-\strategy_{2i})\ +\ \al(\beta'_{2i}-\beta'_{2i-1}) & = & \nal(\strategy_{2i+1}-\beta'_{2i})\ +\ \al(\beta'_{2i+1}-\beta'_{2i}) \\
		\forall i\ \in \{2,\ \ldots\ ,\ k-1\}\\
		\nal(\beta'_{n-1}-\strategy_{n-1})\ +\ \al(\beta'_{n-1}-\beta'_{n-2}) & = & \nal(\beta'_{n-1}-\strategy_n)\ +\ \al(\beta'_n-\beta'_{n-1}).
		\end{eqnarray*}
	}
Solving these equations for $\beta_i$ with $\beta_0=0, \beta_n=1, s_2'= \beta_2$ results in:

	\begin{eqnarray*}
\beta_1 &=&  \frac{1}{2}\beta_2+\frac{1-\alpha}{2}s_1,\\
\beta_2 &=& \frac{\alpha}{1+\alpha}\beta_1 + \frac{\alpha}{1+\alpha}\beta_3+\frac{1-\alpha}{1+\alpha}s_3,\\
\beta_{2i-1} &=& \frac{1}{2}\beta_{2i-2} + \frac{1}{2}\beta_{2i}+\frac{1-\alpha}{2\alpha}s_{2i-1}+\frac{\alpha-1}{2\alpha}s_{2i},\ \forall i \in \{2, \ldots, k-1\},\\
\beta_{2i} &=& \frac{\alpha}{2}\beta_{2i-1} + \frac{\alpha}{2}\beta_{2i+1}+\frac{1-\alpha}{2}s_{2i}+\frac{1-\alpha}{2}s_{2i+1},\ \forall i \in \{2, \ldots, k-1\},\\
\beta_{n-1} &=& \frac{1}{2}\beta_{n-2} + \frac{1}{2}\beta_n+\frac{1-\alpha}{2\alpha}s_{n-1}+\frac{\alpha-1}{2\alpha}s_n.
\end{eqnarray*}
	This system can be solved for $\beta'_1$ and $\beta'_2$ with the help of the Gaussian elimination and generalized continued fractions.
		We apply the Gaussian elimination:
		
			\begin{align*}
		\left(
		\begin{array}{rrrrrrr|lllll}
		1 & -\frac{1}{2} & & & & && \frac{1-\alpha}{2}s_1\\
		-\frac{\alpha}{1+\alpha} & 1 & -\frac{\alpha}{1+\alpha} & & & && \frac{1-\alpha}{1+\alpha}s_3\\
		&-\frac{1}{2}&1&-\frac{1}{2}&&&&\frac{1-\alpha}{2\alpha}s_{2i-1}&+&\frac{\alpha-1}{2\alpha}s_{2i}\\
		&& -\frac{\alpha}{2}&1&-\frac{\alpha}{2}&&&\frac{1-\alpha}{2}s_{2i}&+&\frac{1-\alpha}{2}s_{2i+1}\\
		&&&\ldots&&&&\\
		&&&-\frac{1}{2}&1&-\frac{1}{2}&&\frac{1-\alpha}{2\alpha}s_{2i-1}&+&\frac{\alpha-1}{2\alpha}s_{2i}\\
		&&&& -\frac{\alpha}{2}&1&-\frac{\alpha}{2}&\frac{1-\alpha}{2}s_{2i}&+&\frac{1-\alpha}{2}s_{2i+1}\\
		&&&&&-\frac{1}{2}&1&\frac{1-\al}{2\al}s_{n-1}&+&\frac{\al-1}{2\al}s_n&+&\frac{1}{2}
		\end{array}
		\right).
		\end{align*}
		By adding to the second to last row $\frac\alpha2$ times the last row, our two last rows look as follows:
		
	\begin{align*}
\left(
\begin{array}{ccccrrr|lllll}
&&&& -\frac{\alpha}{2}&1- \frac{\alpha}{4}&0&\frac{1-\alpha}{2}s_{2i}&+&\frac{1-\alpha}{2}s_{2i+1}&+&\frac{\alpha}{4}\\
&&&&&-\frac{1}{2}&1&\frac{1-\al}{2\al}s_{n-1}&+&\frac{\al-1}{2\al}s_n&+&\frac{1}{2}
\end{array}
\right).
\end{align*}
Next, by multiplying the second last row with $\frac{1}{2 \left( 1- \frac\alpha4 \right)}$ and adding it to the last but two, the left side of the three last rows look as follows: 
	\begin{align*}
\left(
\begin{array}{cccccrr|l}
&&&-\frac{1}{2}&1-\frac{\alpha}{4(1-\frac{\alpha}{4})}&0&0&\\
&&&& -\frac{\alpha}{2}&1- \frac{\alpha}{4}&0&\\
&&&&&-\frac{1}{2}&1&
\end{array}
\right).
\end{align*}
We continue this scheme and end up with a left side of the matrix which looks as follows:
			\begin{align*}
{\tiny{\left(
\begin{array}{cccccrr|l}
1 - \frac{\alpha}{2(1+\alpha)(1- \frac{\alpha}{2(1+\alpha)(1+\contFraLin^{n-4})})} & & & &&&&\\
	-\frac{\alpha}{1+\alpha} & 1- \frac{\alpha}{2(1+\alpha)(1+\contFraLin^{n-4})} &&&&&& \\
&\ldots&&&&& \\
&-\frac{1}{2}&1 - \frac{a}{4 (1 - \frac{a}{4 (1 - \frac{a}{4 (1 - \frac{a}{4})})})}&0&0&0&0& \\
&&-\frac{\alpha}{2}&1 - \frac{\alpha}{4 (1 - \frac{\alpha}{4 (1 - \frac{\alpha}{4})})}&0&0&0& \\
&&&& \ldots
\end{array}
\right). }}
\end{align*}
For the right side notice, that for $i \in \{2, \ldots, k-1\}$ $\frac{1-\alpha}{2\alpha}s_{2i-1} + \frac{\alpha-1}{2\alpha}s_{2i}$ is equal to $0$, since both facilities are located at the same position, i.e. $s_{2i-1} = s_{2i}$.
	
	Since we consider facility $2$, we have $\utility_2(\strategyProfile') = \beta'_2 - \beta'_1$.
	Together with $\utility_2(\strategyProfile)=\frac{1}{n}$ and Conjecture \ref{conj:double_highest_improvement} we get
	$
	\rho= \frac{\utility_2(\strategyProfile')}{\utility_2(\strategyProfile)} =  \approxPairingGeneral = n \left(\beta'_2-\beta'_1\right).
	$
\end{proof}

\begin{figure}[htb]
	\centering
	\begin{tikzpicture}
	\begin{axis}[xmin=0, xmax=1, ymin=0.98, ymax=1.3, samples at={0,0.05,...,0.95,1}, width=\columnwidth, height=3.5cm, xlabel={$\alpha$}, ylabel={$\rho$}, ylabel near ticks, xtick={0,0.2,0.4,0.6,0.8,1},legend style={legend columns=3}, x label style={ below=-3mm}]
	\addplot[blue, dashed] {1/4*(-x^2+x+4)};%n=4
	\addplot[red, dashed] {((4+x)*(x*(x*(3+x)-3)-4))/((2+x)*(x*(5*x-2)-8))};%n=5
	\addplot[blue, solid] {(x^3-7*x^2-4*x+16)/(-2*x^2-8*x+16)};%n=6
	\addplot[red, solid] {((64-64*x+7*x^3)*(16+x*(2+x)*(x*(x-3)-2)))/(2*(32+x*(x*(x-10)-16))*(16+x*(x-16+x^2)))}; %n=7
	%\addplot[blue, solid] {-(x^4-17*x^3+24*x^2+48*x-64)/(4*(x^3+x^2-16*x+16))}; %n=8
	\addplot[green, dashed] {((32+(x-4)*x*(2+x)*(1+2*x))*(x*(4+3*x)-16))/((x-2)*(4+x)*(64+x*(x*(x-24)-32)))};%n=9
	\legend{$n=4$, $n=5$, $n \geq 6 $ (even), $n=7$, $n=9$}
	\end{axis}
	\end{tikzpicture}
	\vspace*{-1em} % use for manual adjustment at the end
	\caption{Approximation factor $\rho$ for $\strategyProfile_\text{pair}$ as a function of $\alpha$.}
	\label{fig:facilityApproxFactorsPairingFacilities}
\end{figure}

\noindent Figure~\ref{fig:facilityApproxFactorsPairingFacilities} summarizes the analytically obtained $\rho$-values. The influence of $n$ is negligible for even $n$ with $n\geq 6$. Note, that in contrast to $\strategyProfile_\text{opt}$, the obtained approximation factor is much lower for the facility placement $\strategyProfile_\text{pair}$ with co-located facilities. 

\subsection{Quality of the $\rho$-SPE}

The social costs $\socialcosts(\strategyProfile, \clientStrategyProfile)$ of a strategy profile $(\strategyProfile, \clientStrategyProfile)$ is defined as the sum over the costs of all client agents, i.e., $\socialcosts(\strategyProfile, \clientStrategyProfile) = \int_{\clients} \costsClient(\strategyProfile,\clientStrategyProfile) \text{d}\client$.
Similarly to the Price of Anarchy, we define the quality $\quality$ of an equilibrium as in~\cite{DBLP:journals/corr/FournierS16}.
We are interested in the costs of the client players, while the strategies of the facility players define the stable states.
We define the social optimum of the game as $\text{opt} = \min_{(\strategyProfile, \clientStrategyProfile) \in \strategyProfiles \times \clientStrategyProfiles} \socialcosts(\strategyProfile, \clientStrategyProfile)$.
Then, the quality of an (approximate) pure subgame perfect equilibrium $(\strategyProfile, \clientStrategyProfile)$ is defined as $\quality(\strategyProfile, \clientStrategyProfile) = \frac{\socialcosts(\strategyProfile, \clientStrategyProfile)}{\text{opt}}$.

\begin{lemma}\label{lemma:cost_opt}
	The social optimum $s_{\text{opt}} = (s_1, \ldots, s_n)$ with $s_i = \frac{2i-1}{2n}$ for $i \in \{1, \ldots, n\}$ of the game has $\socialcosts(\strategyProfile_{\text{opt}}, \clientStrategyProfile) = \frac{1+3\alpha}{4n}$.
\end{lemma}

\begin{proof}
	Consider $s_{\text{opt}} = (s_1, \ldots, s_n)$. The interval borders $\beta_i  = \frac{i}{n}$ fulfill for any two neighboring intervals  $\mathcal{I}_i(\strategyProfile) = [\beta_{i-1},\beta_i]$, $\mathcal{I}_j(\strategyProfile) = [\beta_i,\beta_{i+1}]$ the equation $(1-\alpha)  |\strategy_{\choiceFunctions(\strategyProfile,b_i)} - \beta_i| + \alpha  \load_{\choiceFunctions(\strategyProfile,b_i)}(\strategyProfile, \clientStrategyProfile) = (1-\alpha)  |\strategy_{\choiceFunctions(\strategyProfile,a_j)} - \beta_i| + \alpha  \load_{\choiceFunctions(\strategyProfile,a_j)}(\strategyProfile, \clientStrategyProfile)$. So each facility ${\player}$ has the load $\load_{\player}(\strategyProfile, \clientStrategyProfile) = \frac{1}{n}$ and is located in the middle of her corresponding interval. Hence, it follows that $$\socialcosts(\strategyProfile, \clientStrategyProfile) = n \left(\int_0^\frac{1}{n} \left(1-\alpha\right) \left|\frac{1}{2n}-x\right| +\frac{\alpha}{n}\ \text{d}x\right) = \frac{1+3\alpha}{4n}.$$
\end{proof}

\begin{theorem}
	Given a game with $n = 2k$ players for some $k \in \mathbb{N}$ and the state $\strategyProfile_\text{pair} = (s_1, \ldots, s_n)$ with $s_{2 i -1} = s_{2i } = \frac{2i -1}{2k}$ for $i \in \{1, \ldots, k\}$, then $\quality(\strategyProfile_\text{pair}, \clientStrategyProfile) = \frac{2\alpha+2}{3\alpha+1}$.
\end{theorem}

\begin{proof}
	Consider $s = (s_1, \ldots, s_n)$. The interval borders $\beta_i  = \frac{i}{n}$ fulfill for any two neighboring intervals  $\mathcal{I}_i(\strategyProfile) = [\beta_{i-1},\beta_i]$, $\mathcal{I}_j(\strategyProfile) = [\beta_i,\beta_{i+1}]$ the equation $(1-\alpha)  |\strategy_{\choiceFunctions(\strategyProfile,b_i)} - \beta_i| + \alpha  \load_{\choiceFunctions(\strategyProfile,b_i)}(\strategyProfile, \clientStrategyProfile) = (1-\alpha)  |\strategy_{\choiceFunctions(\strategyProfile,a_j)} - \beta_i| + \alpha \load_{\choiceFunctions(\strategyProfile,a_j)}(\strategyProfile, \clientStrategyProfile)$.
	So each facility ${\player}$ has the load $\load_{\player}(\strategyProfile, \clientStrategyProfile) = \frac{1}{n}$ and is located at her interval border. Hence, for the clients' cost it follows $$\socialcosts(\strategyProfile, \clientStrategyProfile) = n \left(\int_0^\frac{1}{n} \left(1-\alpha\right) \left(\frac{1}{n}-x\right) +\frac{\alpha}{n}\ \text{d}x\right) = \frac{1+\alpha}{2n}.$$ With Lemma \ref{lemma:cost_opt} the statement follows.
\end{proof}

\begin{theorem}
	Given a game with $n = 2k-1$ players for some $k \in \mathbb{N}$ and the state $\strategyProfile_\text{pair} = (s_1, \ldots,s_{k-1}, s_k, s_{k+2}, \ldots, s_{n+1})$ with $s_{2 i -1} = s_{2i } = \frac{2i -1}{2k}$ for $i \in \{1, \ldots, k\}$, then $\quality(\strategyProfile_\text{pair}, \clientStrategyProfile) \leq \frac{8(1+\alpha)n^2}{(1+3\alpha)(1+n)^2}$.
\end{theorem}

\begin{proof}
	Consider $s = (s_1, \ldots, s_n)$. First we show that $\load_{\player}(\strategyProfile, \clientStrategyProfile) \geq \frac{1}{n+1}$ for every facility $\player$. 
	We consider the facility $\player$ with the smallest load $\load_{\player}(\strategyProfile, \clientStrategyProfile)$, so $\load_{\player}(\strategyProfile, \clientStrategyProfile) \leq \load_{j}(\strategyProfile, \clientStrategyProfile)$ for $j \neq \player$. Assume there is a facility~$\player$ with $\load_{\player}(\strategyProfile, \clientStrategyProfile) < \frac{1}{n+1}$. For the clients $\client \in [s_i - \frac{1}{n+1}, s_i + \frac{1}{n+1}]$ it holds that $\left| \strategy_\player - \client \right| < \left| \strategy_j - \client \right|$ for all facilities $j$ with $s_i \neq s_j$. Since there is at most one other facility $j$ with $\strategy_j = \strategy_\player$, it follows that there exists a client $\client \in [s_i - \frac{1}{n+1}, s_i + \frac{1}{n+1}]$ with $\costsClient(\strategyProfile,\clientStrategyProfile) = (1-\alpha) \cdot |\strategy_{\choiceFunctions(\strategyProfile,\client)} - \client| + \alpha \cdot \load_{\choiceFunctions(\strategyProfile,\client)}(\strategyProfile, \clientStrategyProfile) > (1-\alpha) \cdot |\strategy_{\player} - \client| + \alpha \cdot \load_{\player}(\strategyProfile, \clientStrategyProfile)$. This contradicts the definition of $\rho$-SPE, since $\client$ can decrease her cost by changing her strategy towards facility $\player$.
	
	It follows for all facilities that  $\load_{\player}(\strategyProfile, \clientStrategyProfile) \leq \frac{2}{n+1}$ and therefore $\left|\strategy_\player - \beta_\player\right| < \frac{2}{n+1}$ and $\left|\strategy_\player - \beta_{\player+1}\right|\\ < \frac{2}{n+1}$, respectively. Hence, it follows for the clients' cost  $$\socialcosts(\strategyProfile, \clientStrategyProfile) = n \left(\int_0^\frac{1}{n+1}\left( 1- \alpha \right) \left(\frac{2}{n+1}-x\right) + \frac{2\alpha}{n+1}\ \text{d}x\right) = \frac{2n(\alpha+1)}{(n+1)^2}.$$ With Lemma \ref{lemma:cost_opt} the statement follows.
\end{proof}

\begin{figure}[htb]
	\centering
	\begin{tikzpicture}
	\begin{axis}[xmin=0, xmax=1, ymin=0.5, ymax=9.5, samples at={0,0.05,...,0.95,1}, width=\columnwidth, height=4.5cm, xlabel={$\alpha$}, ylabel={$\quality(\strategyProfile, \clientStrategyProfile)$}, ylabel near ticks, xtick={0,0.2,0.4,0.6,0.8,1},ytick={1,3,5,7,9}, legend style={legend columns=2}, x label style={ below=-3mm}]
	\addplot[red, solid] {(8*(1+x)*1001^2)/((1+3*x)*(1+1001)^2)}; %n=9
	\addplot[red, dashed] {(8*(1+x)*101^2)/((1+3*x)*(1+101)^2)}; %n=9
	\addplot[orange, solid] {(8*(1+x)*9^2)/((1+3*x)*(1+9)^2)}; %n=9
	\addplot[orange, dashed] {(8*(1+x)*7^2)/((1+3*x)*(1+7)^2)};%n=7
	\addplot[green, solid] {(8*(1+x)*5^2)/((1+3*x)*(1+5)^2)};%n=5 % (8*(1+x)*n^2)/((1+3*x)*(1+n)^2)
	\addplot[blue, solid] {(2*x+2)/(3*x+1)};%n even
	\legend{ $n=1001$, $n=101$,  $n=9$, $n=7$,  $n=5$, $n$ even}
	\end{axis}
	\end{tikzpicture}
	\vspace*{-1em} % use for manual adjustment at the end
	\caption{Quality of the $\rho$-SPE for $\strategyProfile_\text{pair}$ as a function of $\alpha$.}
	\label{fig:facilityQuality}
\end{figure}

\section{Agent-based Simulation}

Kohlberg's model~\cite{KOHLBERG1983211} assumes that the clients are continuously distributed along the linear market, i.e. the interval $[0,1]$ is considered and every point in $[0,1]$ corresponds to a client. This continuous setting is an abstraction from reality and essentially models the case where there are significantly more clients than facilities. Moreover the continuous setting is crucial for our analysis in Section~\ref{sec:analytical} since it enables us to derive analytical results by solving a suitably chosen system of equations. However, as indicated in Section~\ref{sec:analytical}, this approach is tedious to work with and generalizing the results to obtain a closed form solution which depends on $n$ and $\alpha$ seems to be hopeless. In particular, the case for odd $n$ does not yield a symmetric system of equations. Moreover, our proofs cannot be adapted to the discrete version, where we have only a finite number of clients which are spread evenly in the interval $[0,1]$. 

For addressing both problems, the lack of analytical tractability and the transfer of our results to the discrete version, we resort to an agent-based approach. This allows us to derive more general results and to support our conjectures in Section~\ref{sec:analytical}. 
\subsection{Simulation Set-up}
We discretize our model by fixing the total number of clients to some arbitrary value $P$, which we will also call the \emph{precision}. In any discrete instance with exactly $P$ clients we assume that the $P$ clients sit at equally spaced positioned locations in the interval $[0,1]$. More precisely, we assume that the interval $[0,1]$ is subdivided into $P$ consecutive intervals $I_1,\dots,I_P$ of size $\frac{1}{P}$ and that the position of the $i$-th client is the center point $z_i$ of subinterval $I_i$, i.e. $z_i = \frac{i}{p} - \frac{1}{2P}$.

We assume that every client has a weight of $\frac{1}{P}$ and that the total weight assigned to facility $j$ under some client distribution is the sum of the weights of all clients which are assigned to the respective facility and that facilities want to maximize their assigned total weight. Moreover, we assume that facility agents can only select a location from the set $\{z_1,\dots,z_P\}$, i.e. facilities can only be placed on client locations. Note, that if $P\to \infty$ then our discrete model resembles the continuous model. Thus, with increasing precision we can more closely approximate the analytical solution. Our experiments revealed that a precision of $500n$ is sufficient to get very accurate results for $n$ facilities (see Figure~\ref{fig:precision}). 
\begin{figure}[h]
	\centering
	\includegraphics[width=13cm]{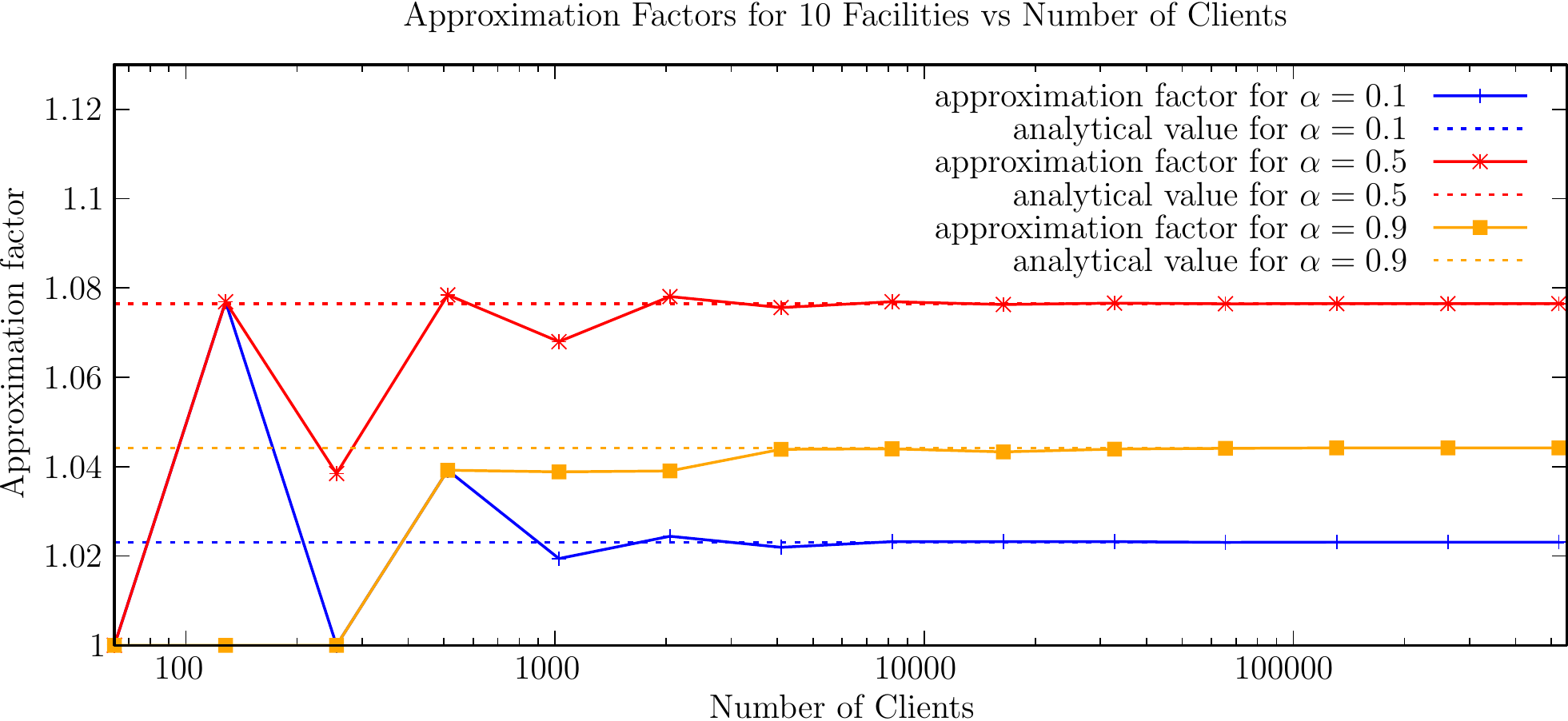}
	\caption{Empirically observed highest approximation factor for the $\rho$-SPE for $n=10$ and $\alpha \in \{0.1,0.5,0.9\}$ compared with its analytical value plotted for increasing precision.}
	\label{fig:precision}
\end{figure}
Moreover, even for fewer numbers of clients, i.e. a lower precision, the obtained results are still very close to the analytical prediction from the continuous model. This emphasizes the value of the continuous model in predicting the behavior of the discrete model.   

\paragraph{Client Simulation} Clients are modeled as selfish autonomous agents which strategically select a facility to minimize their cost. For a given strategy profile $(\strategyProfile,\clientStrategyProfile)$ the cost of client $i$ at position $z_i$ is  
$$C_{z_i}(\strategyProfile,\clientStrategyProfile) = (1-\alpha) \cdot |\strategy_{\choiceFunctions(\strategyProfile,\client_i)} - \client_i| + \alpha \cdot \load_{\choiceFunctions(\strategyProfile,\client_i)}(\strategyProfile, \clientStrategyProfile),$$ where $\load_{\choiceFunctions(\strategyProfile,\client_i)}(\strategyProfile, \clientStrategyProfile) = \sum_{j:\choiceFunctions(\strategyProfile,\client_j) = \choiceFunctions(\strategyProfile,\client_i) } \frac{1}{P}$.

For given fixed facility locations $\strategyProfile = (s_1,\dots,s_n)$, with $s_j \in \{z_1,\dots,z_P\}$ for $1\leq j \leq n$, we invoke round-robin best response dynamics to obtain the \emph{empirical client equilibrium distribution} $D(\strategyProfile)$. There, starting from a fixed initial assignment of clients to facilities, clients are activated in a fixed order and update their strategy with their current best response strategy. By using the discrete analogue of the potential function $\Phi(\strategyProfile,\choiceFunction)$ from the continuous setting, it is straightforward to show that this process converges. Moreover, since the client equilibrium in the continuous setting is unique and since we fix the client activation order and use consistent tie-breaking, the empirical client equilibrium distribution $D(\strategyProfile)$ is unique for any fixed facility placement $\strategyProfile$.  

\paragraph{Facility Simulation} Given a facility placement $\strategyProfile$ and the induced empirical client equilibrium distribution $D(s)$ we compute the best response strategy of a facility $j$ by simply trying all possible locations $z\in \{z_1,\dots,z_P\}$ and computing the induced utility, which equals the load of facility $j$, of each location with the induced empirical client equilibrium distribution $D(z,s_{-j})$. Let $z^*$ denote facility $j$'s best response, then we compute facility $j$'s \emph{highest possible improvement factor} as $\frac{u_j((z^*,s_{-j}),D(z^*,s_{-j}))}{u_j(s,D(s))}$, i.e. the ratio between facility~$j$'s best achievable utility and her current utility.

\paragraph{Computing the Approximation Factor $\rho$} For a given facility placement $\strategyProfile$ and it's corresponding empirical client equilibrium distribution $D(s)$ we obtain the approximation factor $\rho$ of placement $\strategyProfile$ by simply taking the maximum over all facilities of their highest possible improvement factors.

\subsection{Empirical Support for Our Conjectures}
\label{subsec:empirical_support_conjecture}

Our analysis of the continuous model in Section~\ref{sec:analytical}, especially the proofs of Theorems~\ref{theorem:optimum_arbitrary_n} and~\ref{theorem:arbitrary_n_double} crucially relies on Conjectures~\ref{conj:optimum_highest_improvement} and~\ref{conj:double_highest_improvement}, respectively.
While being very challenging to prove analytically, the conjectures can be easily verified with our agent-based approach. For this we compute for a given facility location vector $s \in \{s_{\text{opt}},s_{\text{pair}}\}$ the highest possible improvement factor for each facility (see Figure~\ref{fig:run_plot} for results with $s_{\text{pair}}$).  
\begin{figure}[ht]
	\centering
	\includegraphics[width=13cm]{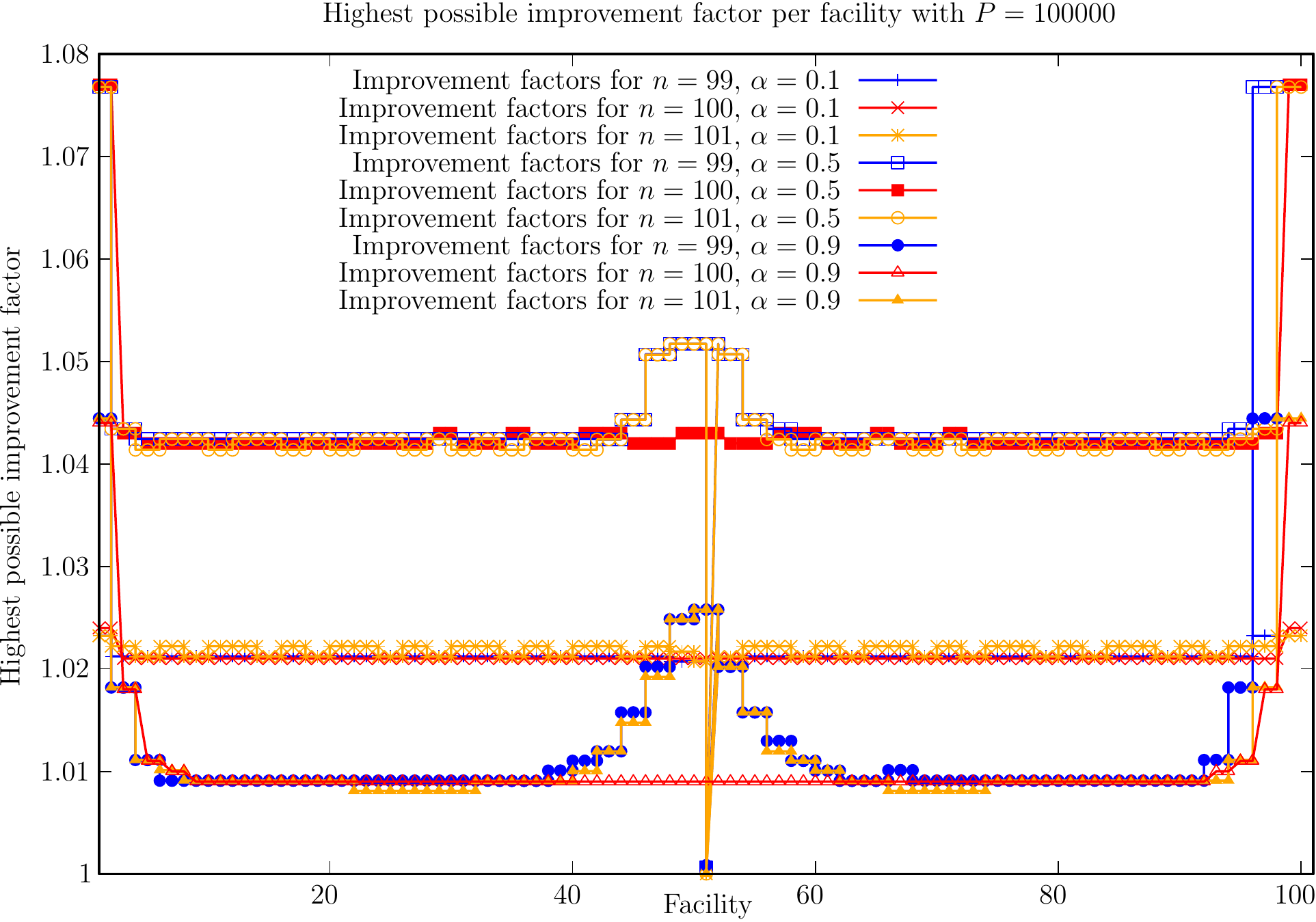}
	\caption{Empirical support for Conjecture~\ref{conj:double_highest_improvement}. Observed improvement factors for each facility for $n \in \{99,100,101\}$ with $P=100000$ and $\alpha \in \{0.1,0.5,0.9\}$ for locations $s_{\text{pair}}$. }
	\label{fig:run_plot}
\end{figure}
We observe that independently from $\alpha$ and $n$ we find that the four outermost facilities which sit at locations $s_1 = s_2$ and $s_{n-1} = s_n$ have the highest improvement factor among all facilities. Moreover our simulations also confirm that the best possible new facility location is the inner border of their assigned client interval, i.e. the location of the client which is assigned to the facility and at the same time has a location as close to $0.5$ as possible. Figure~\ref{fig:run_plot} depicts our obtained results for supporting Conjecture~\ref{conj:double_highest_improvement}. We have similar results regarding Conjecture~\ref{conj:optimum_highest_improvement} but we have to omit them due to space constraints.

\subsection{Worst Approximation Ratio over all $\alpha$}
Finally, we address how the aproximation factor $\rho$ behaves for growing $n$. For this we empirically computed $\rho$ for $n=3$ to $n=100$, where for each $3\leq n \leq 100$ we evaluated every $\alpha$ from $0$ to $1$ in steps of $0.01$. Figure~\ref{fig:peak_plot} shows the maximum approximation factor $\rho$ over all evaluated $\alpha$ for each number of facilities $n$. To avoid numerical issues we scaled $P$ with $n$ as $P = 1000n$. 
\begin{figure}[ht]
	\centering
	\includegraphics[width=13cm]{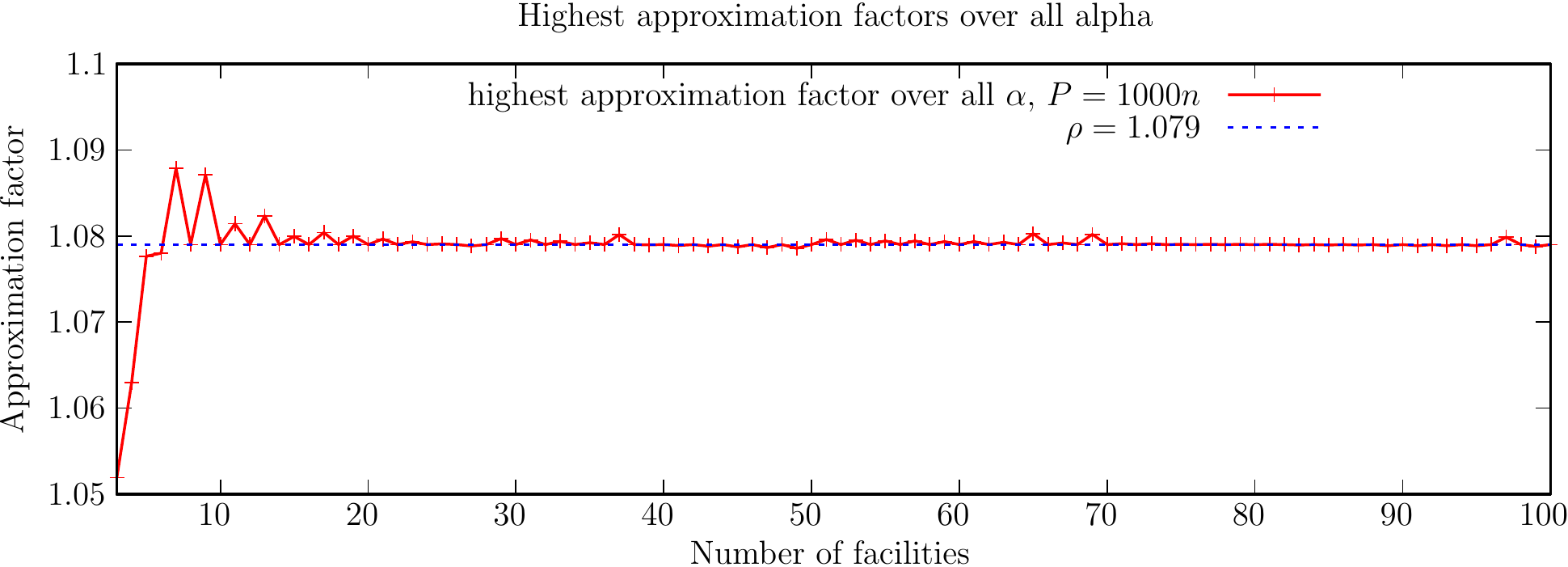}
	\caption{Observed worst approximation ratio for $3\leq n \leq 100$ over all $0\leq \alpha \leq 1$ with precision $P = 1000n$. The maximum $\rho$ approaches $\rho=1.079$ as $n$ increases.}
	\label{fig:peak_plot}
\end{figure}
Our simulation shows that the observed $\rho$ converges to $\rho = 1.079$ as $n$ increases and that the highest approximation factor is obtained for $\alpha = 0.55$. This implies that the investigated approximate subgame perfect equilibria are close to exact equilibria, since the facility agents can only improve their utility by at most $8\%$ by deviating.
\section{Conclusion}
We demonstrated the existence of approximate equilibria with low approximation factors and which adhere to the principle of minimum differentiation for Kohlberg's model. This remarkble contrast to the results of Peters et al.~\cite{Peters2018} indicates that studying approximate equilibria may yield more realistic results than solely focusing on exact equilibria. Moreover, investigating approximate equilibria may also lead to new insights for other models in the realm of Location Analysis.

\bibliographystyle{abbrv}
\bibliography{references}

\end{document}